\documentclass[onecolumn,letterpaper]{IEEEtran}
\IEEEoverridecommandlockouts
\usepackage[utf8]{inputenc} 
\usepackage{soul}
\usepackage[T1]{fontenc} % use cm-super to avoid Type 3 fonts
\usepackage[noadjust]{cite}
\usepackage{amsmath,amssymb,amsfonts,mathtools,bbm}
\usepackage{algorithmic}
\usepackage{textcomp}
\usepackage{xcolor}
\usepackage{amsthm}
\usepackage{tikz}
\usetikzlibrary{decorations.pathreplacing}
\usetikzlibrary{arrows.meta}
\usepackage{hyperref}
\usepackage{graphicx}
\usepackage{balance}
\usepackage[mathscr]{eucal}
\usepackage{enumitem}

\newtheoremstyle{mytheoremstyle}{3pt}{3pt}{\itshape}{}{\bf}{.}{.3em}{} 
\theoremstyle{mytheoremstyle}
\newtheorem{theorem}{Theorem}
\newtheorem{lemma}[theorem]{Lemma}

\newtheorem{proposition}[theorem]{Proposition}

\newtheorem{definition}{Definition}

\newtheorem{construction}{Construction}[section]
\newtheorem{example}{Example}[section]
\theoremstyle{remark}
\newtheorem{remark}{Remark}[section]

\newcommand{\cA}{\mathscr{A}}

\newcommand{\cC}{\mathscr{C}}
\newcommand{\cD}{\mathscr{D}}
\newcommand{\cG}{\mathscr{G}}
\newcommand{\cR}{\mathscr{R}}
\newcommand{\cS}{\mathscr{S}}

\DeclarePairedDelimiter{\norm}{\lVert}{\rVert} 
\newcommand{\etal}{{\em et al.\ }}
\newcommand{\isi}{\mathrm{ISI}}
\newcommand{\de}{d_{\mathrm{edit}}}

\begin{document}
	\title{First-order optimal codes for an adversarial nanopore channel\\
		\thanks{
			\hspace{-1.2em}\rule{1.5in}{.4pt}
			
			H.~Xie is with School of Science and Engineering, The Chinese University of Hong Kong (Shenzhen), Shenzhen, Guangdong, 518172,  P.\,R.\,China (Email: huilingxie@link.cuhk.edu.cn).
			
			Z.~Chen is with School of Science and Engineering, Future Networks of Intelligence Institute, The Chinese University of Hong Kong (Shenzhen), Shenzhen, Guangdong, 518172, P.\,R.\,China (Email: chenztan@cuhk.edu.cn).
			
			%This paper was presented in part at the 2025 IEEE Information Theory Workshop (ITW 2025) \cite{xie2025two} and at the 2026 IEEE International Symposium on Information Theory (ISIT 2026) \cite{xie2026deletion}.
		}
	}
	\author{%
		{Huiling Xie} 
		\hspace*{1in}\and
		{Zitan Chen}
		% \IEEEauthorblockA{
			% 	\IEEEauthorrefmark{1}\IEEEauthorrefmark{2}School of Science and Engineering\\
			% 	\IEEEauthorrefmark{2}Future Networks of Intelligence Institute\\
			% 	\IEEEauthorrefmark{1}\IEEEauthorrefmark{2}The Chinese University of Hong Kong, Shenzhen\\
			% 	Email: chenztan@cuhk.edu.cn}
	}
	\maketitle
	
	\begin{abstract}
		We study error-correcting codes for an adversarial nanopore channel, where a $q$-ary string is first transformed by an inter-symbol interference channel with window size $\ell$ into a sequence of overlapping $\ell$-mers, and an adversary then corrupts this $\ell$-mer sequence by introducing at most $t$ edits.
		For the deletion-only nanopore channel, we show that the optimal redundancy of $t$-deletion-correcting codes of length $n$ lies between $t\log_q n+\Omega(1)$ and $2t\log_q n-\log_q\log_2 n+O(1)$.
		We then give two explicit deletion-correcting constructions in the regime $t\leq \min\{(\ell-1)/2,(\ell+2)/3\}$.
		The first construction relies on generalized Reed-Solomon codes and has redundancy $2t\log_q n+\Theta(\log\log n)$.
		The second is based on Sidon sets (or rather $B_t$ sequences) and has redundancy $t\log_q n+\Theta(\log\log n)$, matching the lower bound to first order.
		We further extend the $B_t$-based approach to the edit channel, allowing insertions, deletions, and substitutions of $\ell$-mers. In the regime $t\leq \min\{(\ell-1)/4,(\ell+2)/6\}$, this gives explicit $t$-edit-correcting codes with redundancy $t\log_q n+\Theta(\log\log n)$, which is first-order optimal.
	\end{abstract}
	
	\section{Introduction}
	Nanopore sequencing has attracted considerable attention for DNA-based data storage due to its advantages over other sequencing technologies. However, it is prone to several types of errors, including deletion, duplication and substitution errors. In essence, nanopore sequencing can be modeled as passing a string through an inter-symbol interference (ISI) channel, followed by additional noisy channels that further corrupt the ISI channel output. Specifically, the ISI channel scans the string from left to right using a sliding window of size $\ell$, shifted by one position each step, thereby producing a sequence of $\ell$-mers. The subsequent noisy channels then introduce errors into this $\ell$-mer sequence. The goal of nanopore sequencing is to recover the original string from the final noisy output.
	
	A large body of work has studied nanopore sequencing in recent years. Depending on whether the underlying channel model is probabilistic, the literature can be grouped into two broad classes. Under a probabilistic channel model, Magner \etal \cite{magner2016fundamental} gave an information-theoretic analysis of the impact of sticky insertion-deletion errors in nanopore sequencers, without explicitly accounting for ISI. Later, Mao \etal \cite{mao2018models} developed a signal-level model for nanopore sequencing that captures ISI and analyzed achievable information rates. In \cite{hamoum2021channel,hamoum2023channel,maarouf2023achievable}, the authors studied a nanopore channel model based on a Markov chain with memory of order $\ell$. In subsequent work, McBain \etal \cite{mcbain2022finite,mcbain2024information} introduced a channel model in the language of semi-Markov processes and estimated its achievable information rates using efficient algorithms. The achievable information rates of this model were further examined in \cite{arvind2025achievable,mcbain2025achievable}. A number of other works have also considered the design of coding schemes for nanopore sequencing under various probabilistic models. We refer the interested reader to the recent works \cite{mcbain2023homophonic,vidal2024concatenated,volkel2025nanopore,whritenour2025constrained,lin2025block} and the references therein.	
	
	Another line of research assumes deterministic channel models and investigates the combinatorial and algebraic coding problems arising in nanopore sequencing. Specifically, the authors of \cite{hulett2021coding} introduced a deterministic model capturing only ISI and substitution errors, and proposed coding schemes for error correction. Extending the work of \cite{hulett2021coding}, the authors of \cite{chee2024coding} studied coding schemes that can deal with duplication and deletion errors. In addition, a number of recent papers model the noisy output of deterministic nanopore channels as a sequence of erroneous $\ell$-mer compositions. 
	Under this framework, the work \cite{banerjee2024error} considered ISI effect with substitution errors and gave a coding scheme that can correct a single substitution with near-optimal redundancy. Extensions to multiple substitution errors were further explored in \cite{sun2024bounds,banerjee2025correcting}.
	Assuming instead that only deletion errors are present in addition to ISI, the work \cite{banerjee2024correcting} proposed a single-deletion-correcting code with near-optimal redundancy. 
	Meanwhile, the paper \cite{yu2025asymptotic} studied ISI with tandem-duplication errors in the $\ell$-mer composition sequence and analyzed the optimal code rate.
	We also note that the recent works \cite{chee2024transverse,yerushalmi2024capacity} also studied various coding-theoretic questions under the assumption that the noisy channel output is in the form of $\ell$-mer compositions.
	
	Motivated by the fact that basecallers are often used to infer a sequence of $\ell$-mers from the output of practical nanopore sequencers, we model the noisy channel output as a sequence of erroneous $\ell$-mers and
	study bounds and constructions for codes that can correct a constant number of $\ell$-mer edits,
	allowing insertions, deletions, and substitutions of $\ell$-mers. 
	We note that modeling the noisy channel output as a sequence of $\ell$-mers yields a model that resembles the $\ell$-symbol read channel \cite{cassuto2011codes,yaakobi2016constructions,chee2021locally}. However, the reads in the $\ell$-symbol read channel are taken cyclically. That is, after reading the last $\ell$ symbols of the input, the channel continues with the size-$\ell$ window that wraps around the end of the string, ending with the window consisting of the last symbols followed by the first $\ell-1$ symbols. In \cite{chee2021locally}, Chee \etal introduced locally-constrained de Bruijn codes for correcting synchronization errors in $\ell$-symbol read channels. For the case of $\ell=2$, codes correcting deletion and substitution errors were studied in \cite{chee2020codes}.
	
	The nanopore channel considered in this work is also closely related to the classical adversarial edit channel, and in particular to the $q$-ary deletion channel \cite{levenshtein1966binary}. For the binary deletion channel, Levenshtein \cite{levenshtein1966binary} showed that the optimal redundancy for $t$-deletion-correcting binary codes of length $n$ lies between $t\log_2 n+o(\log n)$ and $2t\log_2 n+o(\log n)$, where $t$ is a constant independent of $n$. For a single deletion, the Varshamov-Tenengolts (VT) construction \cite{levenshtein1966binary,varshamov1965codes} is known to have redundancy at most $\log_2(n+1)$, and for two deletions, Guruswami and H{\aa}stad \cite{guruswami2021explicit} constructed codes with redundancy $4 \log_2 n + O(\log \log n)$, matching the existential upper bound. More generally, explicit constructions of $t$-deletion-correcting binary codes with order-optimal redundancy are presented in \cite{sima2020optimal,cheng2018deterministic}. In particular, the construction of Sima and Bruck \cite{sima2020optimal} achieves redundancy $8t\log_2 n+o(\log n)$. For $q$-ary codes, Levenshtein \cite{levenshtein2002bounds} proved analogous redundancy bounds $t\log_q n + o(\log n)$ and $2t\log_q n +o(\log n)$ (measured in $q$-ary symbols). A construction with nearly optimal redundancy for a single deletion was given by Tenengolts \cite{tenengolts1984nonbinary}. In \cite{sima2020optimal2}, Sima \etal proposed $t$-deletion-correcting $q$-ary codes that require at most $4t(1+\epsilon)\log_q n+o(\log n)$ redundant symbols. On the existential side, Alon \etal \cite{alon2023logarithmically} obtained the first asymptotic improvement to Levenshtein's lower bound on the optimal code size, while Levenshtein's upper bound has been improved in various parameter regimes in a series of works \cite{kulkarni2013nonasymptotic,fazeli2015generalized,cullina2016generalized,yasunaga2024improved,liu2022bounds,kong2025combinatorial}. 
	
	Compared with the deletion-only setting, explicit constructions for the general $q$-ary edit channel, allowing insertions, deletions, and substitutions, are less developed. Levenshtein \cite{levenshtein1966binary} showed that the binary VT construction can correct a single edit with redundancy $\log_2(n+1)$, and later works \cite{cai2021correcting,gabrys2022beyond} constructed single-edit-correcting codes with nearly optimal redundancy over non-binary	alphabets. For multiple edits, one standard reduction is that any $2t$-deletion-correcting code is also $t$-edit-correcting, since a substitution may be viewed as a deletion followed by an insertion. This reduction, however, generally doubles the deletion budget and does not yield optimal redundancy. Conversely, since deletions form a special case of edits, any lower bound on the redundancy of $t$-deletion-correcting codes also applies to $t$-edit-correcting codes.
	
	\indent{\em Our contributions.} 
	We first study the deletion-only nanopore channel. We observe that this channel is weaker than the classical adversarial deletion channel. Consequently, the lower bound of Alon \etal \cite{alon2023logarithmically} applies directly and yields an existential upper bound $2t\log_q n-\log_q\log_2 n+O(1)$ on the optimal redundancy. 
	Moreover, drawing on ideas of Levenshtein \cite{levenshtein1966binary}, we use the probabilistic method to establish an upper bound on the optimal code size for the deletion-only nanopore channel, implying a lower bound of $t\log_q n+\Omega(1)$ on the optimal redundancy.
	
	We then present two explicit constructions of $t$-deletion-correcting codes in the small-$t$ regime where $t\leq \min\{\ell-1)/2,(\ell+2)/3\}$. 
	The first construction reduces deletion correction to an almost Hamming-type error-correction problem and uses generalized Reed-Solomon (GRS) codes as a key building block.\footnote{More generally, any efficiently decodable $t$-Hamming-error-correcting code over a suitable alphabet would suffice for this purpose.} It gives $q$-ary codes of length $n$ with redundancy $2t\log_q n+\Theta(\log\log n)$.
	The second construction exploits the fact that the relevant error pattern that GRS codes have to deal with is a nonnegative integer vector whose total weight is at most $t$. Using Sidon sets, more precisely $B_t$-sequences, we obtain $t$-deletion-correcting codes with redundancy $t\log_q n+\Theta(\log\log n)$, which matches the lower bound $t\log_q n+\Omega(1)$ to first order.
	
	Lastly, we extend the $B_t$-based construction to the general edit nanopore channel, where the adversary may insert, delete, or substitute $\ell$-mers. In the regime $t\leq \min\{(\ell-1)/4,(\ell+2)/6\}$, we construct explicit $t$-edit-correcting codes with redundancy $t\log_q n+\Theta(\log\log n)$, again achieving first-order optimal redundancy.

	\indent{\em Paper outline.} The rest of the paper is organized as follows. Section~\ref{sec:model} formally introduces our channel model and present several structural results used throughout the paper.  Section~\ref{sec:bounds} establishes existential bounds on the optimal code size for the deletion-only nanopore channel. 
	Section~\ref{sec:con} gives explicit constructions of $t$-deletion-correcting codes for this channel.
	Section~\ref{sec:edit} presents constructions of $t$-edit-correcting codes. 
	Finally, Section~\ref{sec:re} discusses open problems and future directions.
	
	\section{Channel model and structural results}~\label{sec:model}
	Let \(\Sigma_q\) be an alphabet of size \(q\). While DNA sequencing corresponds to the case $q=4$ with \(\Sigma_4 = \{\mathtt{A}, \mathtt{T}, \mathtt{C}, \mathtt{G}\}\), we consider the general setting in which $\Sigma_q$ is any finite alphabet of size $q>1$. 
	Let \(\mathbf{x}:=(x_1,\dots,x_n)\in\Sigma_q^n\) be a $q$-ary string of length $n$ to be sequenced. 
	The nanopore channel considered in this work is a concatenation of an ISI channel with window size $\ell>1$ and an adversarial edit channel that imposes at most $t$ edits, where an edit is an insertion, deletion or substitution.
	Throughout, we assume $\ell$ and $t$ are constants independent of $n$.
	Motivated by practical applications of nanopore sequencing, where known adapter sequences are attached to both ends of the DNA strand to be sequenced, we assume that two (possibly distinct) length‑$\ell$ strings over $\Sigma_q$ are added to the string $\mathbf{x}$, one as a prefix and the other as a suffix.\footnote{We note that adapter sequences used in practice may be much longer than $\ell$, but here we assume they have length $\ell$ for simplicity.}
	In other words, the ISI channel (and hence the nanopore channel) takes \(\tilde{\mathbf{x}}:=(x_{-\ell+1},\dots,x_0,\mathbf{x},x_{n+1},\dots,x_{n+\ell})\in\Sigma_q^{n+2\ell}\) as the \emph{input string}.
	More precisely, the ISI channel is a mapping that maps the input string $\tilde{\mathbf{x}}\in\Sigma_q^{n+2\ell}$ to a sequence of its substrings $\mathbf{s}_0^{n+\ell}:=(\mathbf{s}_0,\dots,\mathbf{s}_{n+\ell})\in (\Sigma_q^\ell)^{n+\ell+1}$, where for \(i = 0, \dots, n+\ell\), the substring \(\mathbf{s}_i := (x_{i-\ell+1}, \dots, x_i)\) is called an $\ell$-mer. Clearly, the $\ell$-mers $\mathbf{s}_0$ and $\mathbf{s}_{n+\ell}$ correspond to the known adapter sequences.
	The edit channel then takes the sequence of $\ell$-mers \(\mathbf{s}_0^{n+\ell}\) as input and apply at most \(t\) edits to these \(\ell\)-mers.
	In this work, we adopt the anchored convention that the two adapter $\ell$-mers are not edited and that no insertions occur before $\mathbf{s}_0$ or after $\mathbf{s}_{n+\ell}$. Thus, the output of the edit channel has the form $\mathbf{z}_0^N :=(\mathbf{z}_0,\ldots,\mathbf{z}_N)\in(\Sigma_q^\ell)^{N+1}$, where $N+1$ is the number of output $\ell$-mers, and $\mathbf{z}_0=\mathbf{s}_0, \mathbf{z}_N=\mathbf{s}_{n+\ell}$. The block diagram of the nanopore channel is depicted in Figure~\ref{fig:channel}.
	
	\tikzstyle{startstop} = [rectangle, rounded corners, minimum height=0.8cm, text centered, draw = white]
	\tikzstyle{process} = [rectangle, minimum width=2cm, minimum height=0.8cm, text centered, draw=black]
	\begin{figure}[!t]
		\centering
		\begin{tikzpicture}[node distance=1.5cm]
			\node[coordinate] (start) {};
			\node (p1) [process] [below of=start] {ISI channel with window size $\ell$};
			\node (p2) [process] [below of=p1] {Adversarial channel with at most $t$ edits};
			\node[coordinate, below of=p2] (end) {};
			
			\draw[->] (start) -- node [right] {$\tilde{\mathbf{x}} = (x_{-\ell+1}, \dots, x_{n+\ell})\in \Sigma^{n+2\ell}$} (p1)  ;
			\draw [->] (p1) -- node [right] {$\mathbf{s}_0^{n+\ell} = (\mathbf{s}_0, \dots, \mathbf{s}_{n+\ell})\in(\Sigma_q^{\ell})^{n+\ell+1}$} (p2)  ;
			\draw [->] (p2) -- node [right] {$\mathbf{z}_0^{N} = (\mathbf{z}_{0}, \dots, \mathbf{z}_{N})\in(\Sigma_q^{\ell})^{N+1}$} (end);
		\end{tikzpicture}
		\caption{The adversarial nanopore channel model}
		\label{fig:channel}
	\end{figure}
	
	The central question studied in this work is the recovery of the $q$-ary string $\mathbf{x}$ from the sequence of $\ell$-mers $\mathbf{z}_0^{N}$. 
	Since any two consecutive $\ell$-mers in $\mathbf{s}_0^{n+\ell}$ overlap in $\ell-1$ symbols due to the ISI effect, a natural first step for reconstructing $\mathbf{x}$ is to recover a string $\tilde{\mathbf{y}}$ that is an edited version of $\tilde{\mathbf{x}}$ from $\mathbf{z}_0^{N}$ by utilizing the overlaps. This can be achieved by detecting ``gaps'' in $\mathbf{z}_0^{N}$ and interpolating $\ell$-mers to fill the ``gap'' between adjacent $\ell$-mers. More precisely, for each $i=0,\ldots,N-1$, an integer $\ell_i\in\{1,\ldots,\ell\}$ is called a \emph{gap} between $\mathbf{z}_i$ and $\mathbf{z}_{i+1}$ if the length-$(\ell-\ell_i)$ suffix of $\mathbf{z}_i$ is the same as the length-$(\ell-\ell_i)$ prefix of $\mathbf{z}_{i+1}$.
	The case $\ell_i=1$ corresponds to the ISI overlap of length $\ell-1$, and we call two consecutive $\ell$-mers $\mathbf{z}_{i}$ and $\mathbf{z}_{i+1}$ \emph{consistent} if $\ell_i=1$. A gap between two $\ell$-mers need not be unique. For instance, if the symbols in $\mathbf{z}_{i}$ and $\mathbf{z}_{i+1}$ are all the same, then $\ell_i$ may take any integer between $1$ and $\ell$. On the other hand, there always exists a gap between $\mathbf{z}_{i}$ and $\mathbf{z}_{i+1}$ since $\ell_i=\ell$ imposes no overlap constraint. For $i=0,\ldots,N-1$, define $\ell_i^*$ to be the smallest gap between $\mathbf{z}_i$ and $\mathbf{z}_{i+1}$.
	By detecting the smallest gaps in the any two consecutive $\ell$-mers in $\mathbf{z}_0^N$, we can insert a minimum number of $\ell$-mers to obtain a sequence from $\mathbf{z}_0^N$ in which any adjacent $\ell$-mers are consistent.
	
	\begin{proposition}
		\label{le:consistent}
		For each $i\in\{0,\ldots,N-1\}$, inserting $\ell_i^*-1$ $\ell$-mers between $\mathbf{z}_i$ and $\mathbf{z}_{i+1}$ gives a shortest $\ell$-mer sequence in which any two consecutive $\ell$-mers are consistent.
	\end{proposition}
	\begin{proof}
		Denote by $\mathbf{t}_{i+1}$ the length-$\ell_i^*$ suffix of $\mathbf{z}_{i+1}$. Then the string $(\mathbf{z}_i,\mathbf{t}_{i+1})$ has $\mathbf{z}_i$ as its first $\ell$-mer and $\mathbf{z}_{i+1}$ as its last $\ell$-mer. Its intermediate length-$\ell$ substrings therefore form a consistent $\ell$-mer sequence between $\mathbf{z}_i$ and $\mathbf{z}_{i+1}$, and the number of these intermediate $\ell$-mers is $\ell_i^*-1$. 
		Minimality follows from the definition of $\ell_i^*$.
	\end{proof}
	
	\begin{example}\label{ex:prop1}
		Assume $n=10$, $\ell=4$ and $t=3$. Let $\tilde{\mathbf{x}}=\mathtt{ATGC\ ACTTTGTCCA\ TTGC}$ be the input string, so the two adapter sequences are $\mathbf{s}_0=\mathtt{ATGC}$ and $\mathbf{s}_{14}=\mathtt{TTGC}$. The ISI channel maps $\tilde{\mathbf{x}}$ to
		\[
		\mathbf{s}_0^{14}
		=
		\left(
		\begin{array}{cccccccc}
			\mathtt{ATGC} & \mathtt{TGCA} & \mathtt{GCAC} & \mathtt{CACT} &
			\mathtt{ACTT} & \mathtt{CTTT} & \mathtt{TTTG} & \mathtt{TTGT} \\
			\mathbf{s}_0 & \mathbf{s}_1 & \mathbf{s}_2 & \mathbf{s}_3 &
			\mathbf{s}_4 & \mathbf{s}_5 & \mathbf{s}_6 & \mathbf{s}_7
		\end{array}
		\right.
		\]
		\[
		\left.
		\begin{array}{ccccccc}
			\mathtt{TGTC} & \mathtt{GTCC} & \mathtt{TCCA} & \mathtt{CCAT} &
			\mathtt{CATT} & \mathtt{ATTG} & \mathtt{TTGC} \\
			\mathbf{s}_8 & \mathbf{s}_9 & \mathbf{s}_{10} & \mathbf{s}_{11} &
			\mathbf{s}_{12} & \mathbf{s}_{13} & \mathbf{s}_{14}
		\end{array}
		\right).
		\]
		Suppose the edit channel imposes $t=3$ edits on the input $\ell$-mers $\mathbf{s}_0^{14}$: it deletes $\mathbf{s}_5=\mathtt{CTTT}$ and $\mathbf{s}_6=\mathtt{TTTG}$, and substitutes $\mathbf{s}_{10}=\mathtt{TCCA}$ with $\mathtt{GTAT}$. Therefore, the channel output $\ell$-mers are
		\[
		\mathbf{z}_0^{12}
		=
		\left(
		\begin{array}{ccccccc}
			\mathtt{ATGC} & \mathtt{TGCA} & \mathtt{GCAC} & \mathtt{CACT} &
			\mathtt{ACTT} & \mathtt{TTGT} & \mathtt{TGTC} \\
			\mathbf{z}_0 & \mathbf{z}_1 & \mathbf{z}_2 & \mathbf{z}_3 &
			\mathbf{z}_4 & \mathbf{z}_5 & \mathbf{z}_6
		\end{array}
		\right.
		\]
		\[
		\left.
		\begin{array}{cccccc}
			\mathtt{GTCC} & \mathtt{GTAT} & \mathtt{CCAT} & \mathtt{CATT} &
			\mathtt{ATTG} & \mathtt{TTGC} \\
			\mathbf{z}_7 & \mathbf{z}_8 & \mathbf{z}_9 & \mathbf{z}_{10} &
			\mathbf{z}_{11} & \mathbf{z}_{12}
		\end{array}
		\right).
		\]
		The smallest gaps between consecutive output $\ell$-mers are $(\ell_0^*,\ldots,\ell_{11}^*)=(1, 1, 1, 1, 2, 1, 1, 4, 4, 1, 1, 1)$. Note that only $\ell^*_4$, $\ell^*_7$ and $\ell^*_8$ are larger than $1$.
		Thus, to obtain a sequence of consistent $\ell$-mers, insertions are needed only between the pairs
		\[
		(\mathbf{z}_4,\mathbf{z}_5),\qquad
		(\mathbf{z}_7,\mathbf{z}_8),\qquad
		(\mathbf{z}_8,\mathbf{z}_9).
		\]
		The length-$\ell^*_4$ suffix of $\mathbf{z}_5$ is $\mathbf{t}_5=\mathtt{GT}$. So $(\mathbf{z}_4, \mathbf{t}_5)=\mathtt{ACTTGT}$ and its only intermediate $\ell$-mer is $\mathtt{CTTG}$, which can be inserted between $\mathbf{z}_4$ and $\mathbf{z}_5$. 
		Similarly, we obtain $(\mathbf{z}_7, \mathbf{t}_8) = \mathtt{GTCCGTAT}$ and $(\mathbf{z}_8, \mathbf{t}_9) = \mathtt{GTATCCAT}$. So the interpolated $\ell$-mers between $\mathbf{z}_7$ and $\mathbf{z}_8$ are $\mathtt{TCCG, CCGT, CGTA}$, and the interpolated $\ell$-mers between $\mathbf{z}_8$ and $\mathbf{z}_9$ are $\mathtt{TATC, ATCC, TCCA}$. 
		After inserting these interpolated $\ell$-mers, we obtain the following consistent $\ell$-mer sequence of length $20$:
		\[
		\begin{aligned}
			(&\mathtt{ATGC},\mathtt{TGCA},\mathtt{GCAC},\mathtt{CACT},\mathtt{ACTT},
			\mathtt{CTTG},\mathtt{TTGT},\mathtt{TGTC},\\
			&\mathtt{GTCC},\mathtt{TCCG},\mathtt{CCGT},\mathtt{CGTA},\mathtt{GTAT},
			\mathtt{TATC},\mathtt{ATCC},\\
			&\mathtt{TCCA},\mathtt{CCAT},\mathtt{CATT},\mathtt{ATTG},\mathtt{TTGC}),
		\end{aligned}
		\]
		This example also illustrates that the interpolated $\ell$-mers need not coincide with the $\ell$-mers actually deleted by the channel.
	\end{example}
	
	Applying Proposition~\ref{le:consistent} to $\mathbf{z}_0^N$, we obtained a shortest consistent completion of the output $\ell$-mers. The resulting sequence corresponds to the $\ell$-mers of the $q$-ary string $\tilde{\mathbf{y}}:=(\mathbf{z}_0,\mathbf{t}_{1},\ldots,\mathbf{t}_{N})$ of length $\ell+\sum_{i=0}^{N-1}\ell_i^*$, where for $i=0,\ldots,N-1$, the string $\mathbf{t}_{i+1}$ is the length-$\ell_i^*$ suffix of $\mathbf{z}_{i+1}$.
	This string $\tilde{\mathbf{y}}$ can then be viewed as an edited version of the input string $\tilde{\mathbf{x}}$.
	For simplicity, we denote the mapping from the output $\ell$-mers $\mathbf{z}_0^{N}$ to the string $\tilde{\mathbf{y}}$ by $\psi$ and refer to $\tilde{\mathbf{y}}=\psi(\mathbf{z}_0^{N})$ as the \emph{output string} of the nanopore channel corresponding to the input string $\tilde{\mathbf{x}}$.
	
	It is clear that any inconsistency in $\mathbf{z}_0^{N}$ must result from edits introduced by the channel. Conversely, nontrivial edits may preserve consistency.
	We next characterize the local structure that permits this to happen. The key observation is that such nontrivial edits force periodicity.
	
	\begin{definition}[Periodic patterns]\label{def:periodic}
		A string $\mathbf{s}=(s_1,\ldots,s_n)\in\Sigma_q^n$ is said to have period $T\leq n-1$ if $s_i=s_{i+T}$ for all $i=1,\ldots,n-T$.
		The least period of $\mathbf{s}$ is the smallest positive integer $T$ with this property.
		A substring of \({\mathbf{s}}\) is called a periodic pattern with period $T$ if it has period $T$.
		A periodic pattern with period $T$ in $\mathbf{s}$ is called maximal if it cannot be extended to the left or right within $\mathbf{s}$ to form a longer substring that has the same period $T$.
	\end{definition}
	
	\begin{remark}
		Throughout, a string of length $n$ is allowed to have a	period $T>n/2$. In this case, the string consists of one full period followed by a prefix of the next period. 
	\end{remark}	
	The following theorem of Fine and Wilf \cite{fine1965uniqueness} is helpful for our subsequent discussion on periodic patterns.
	
	\begin{theorem}[{Fine-Wilf theorem \cite[Theorem~1]{fine1965uniqueness}}]\label{thm:FW}
		Let $(f_i)_{i\geq 0},(g_i)_{i\geq 0}$ be two strings over $\Sigma_q$ with period $T_1,T_2$, respectively. If $f_i=g_i$ for $T_1+T_2-\gcd(T_1,T_2)$ consecutive integer $i$, then $f_i=g_i$ for all $i$. The result would be false if $T_1+T_2-\gcd(T_1,T_2)$ were replaced by anything smaller.  
	\end{theorem}
	
	\begin{remark}\label{re:FW}
		One direct implication of Theorem~\ref{thm:FW} is that, within a string, two adjacent maximal periodic patterns with period $T_1,T_2$ can overlap in at most $T_1+T_2-\gcd(T_1,T_2)-1$ positions. It also implies that if the length of a maximal periodic pattern is at most $T_1+T_2-\gcd(T_1,T_2)-1$, it is possible that it is contained in another periodic pattern. 
	\end{remark}
	
	\begin{remark}\label{re:FW2}
		If a periodic pattern with period $T$ has length at least $2T-2$, then its least period must divide $T$. This is because if $T'\leq T$ is the least period, then $2T-2\geq T+T'-\gcd(T,T')$. Therefore, by Theorem~\ref{thm:FW}, the pattern also has period $\gcd(T,T')$, which divides $T'$. But $T'$ is the least period so $\gcd(T,T')=T'$ and $T'|T$.
	\end{remark}

	The next lemma describes when a single edit burst\footnote{In this paper, a edit burst always refers to a maximal edit burst.} can preserve consistency.
	
	\begin{lemma}
		\label{le:consistent-edit}
		Let $\tau,\rho \geq 0$ be integers. 
		Let a single edit burst replace $\mathbf{s}_{i+1},\ldots,\mathbf{s}_{i+\tau}$ by $\ell$-mers $\mathbf{e}_1,\ldots,\mathbf{e}_\rho$ while leaving the two boundary $\ell$-mers $\mathbf{s}_i$ and $\mathbf{s}_{i+\tau+1}$ unchanged. 
		Define $\mu=\min\{\tau,\rho\}+1$ and $\delta=|\tau-\rho|$.
		Assume that the local output $\ell$-mers $\mathbf{s}_i,\mathbf{e}_1,\ldots,\mathbf{e}_\rho,
		\mathbf{s}_{i+\tau+1}$ are consistent. 
		\begin{enumerate}[label=(\roman*)]
			\item If $\tau=\rho\leq \ell-1$, then the edit burst is trivial: $\mathbf{e}_j=\mathbf{s}_{i+j}$ for $j=1,\ldots,\rho$.
			\item\label{ite:nontrivial-edit} If $\tau\neq \rho$ and $\mu\leq \ell-1$, then
			the longer of the local input and output strings contains a periodic pattern of length $\delta+\ell-\mu$ with period $\delta$. More precisely, if $\tau>\rho$, this periodic pattern lies in	the input string; if $\tau<\rho$, it lies in the output string.
			Moreover, if $\delta\leq \ell-\mu+2$, then the least period of this periodic pattern divides $\delta$.
		\end{enumerate}
	\end{lemma}
	
	\begin{proof}
		First assume $\tau=\rho\leq\ell-1$. Then the two boundary $\ell$-mers $\mathbf{s}_i,\mathbf{s}_{i+\tau+1}$ of the local output $\ell$-mers $\mathbf{s}_i,\mathbf{e}_1,\ldots,\mathbf{e}_\rho,
		\mathbf{s}_{i+\tau+1}$, together with their relative position shift $\tau+1\leq \ell$, determine a unique local string, formed by $\mathbf{s}_i$ and the length-$(\tau+1)$ suffix of $\mathbf{s}_{i+\tau+1}$.
		This is exactly the local input string corresponding to the local input $\ell$-mers $\mathbf{s}_i,\mathbf{s}_{i+1}\ldots,\mathbf{s}_{i+\tau},\mathbf{s}_{i+\tau+1}$, so the intermediate output $\ell$-mers must satisfy $\mathbf{e}_j=\mathbf{s}_{i+j}$ for all $j=1,\ldots,\rho$. 
		
		Now assume $\tau\neq\rho$ and $\mu\leq\ell-1$. If $\tau>\rho$, then $\rho\leq \ell-2$, and consistency of the local output $\ell$-mers implies that the length-$(\ell-\rho-1)$ suffix of $\mathbf{s}_i$ equals the length-$(\ell-\rho-1)$ prefix of $\mathbf{s}_{i+\tau+1}$, since the relative position shift of the boundary $\ell$-mers in the local output $\ell$-mers is $\rho+1$.
		Since $\mathbf{s}_i=(x_{i-\ell+1},\ldots,x_i)$ and $\mathbf{s}_{i+\tau+1}=(x_{i-\ell+\tau+2},\ldots,x_{i+\tau+1})$, this further implies
		\[
		x_j=x_{j+(\tau-\rho)}=x_{j+\delta},\qquad j=i-\ell+\rho+2,\ldots,i.
		\]
		Therefore, the input string contains the substring $(x_{i-\ell+\rho+2},\ldots,x_{i+\delta})$, which has length $\delta+\ell-\rho-1=\delta+\ell-\mu$ and period $\delta$. 
		The case $\tau<\rho$ is symmetric: using the relative position shift $\tau+1$ of the boundary $\ell$-mers in the local input $\ell$-mers, one can show that the corresponding periodic pattern lies in the output string and has length $\delta+\ell-\tau-1=\delta+\ell-\mu$.
		
		It remains to prove the divisibility statement.
		Since $\delta \leq \ell-\mu +2$, i.e., $\delta-2\leq \ell-\mu$, the length of the periodic pattern satisfies
		\[
		\delta+\ell-\mu\geq 2\delta-2.
		\]
		Therefore, by Remark~\ref{re:FW2}, the least period of the periodic pattern divides $\delta$.
	\end{proof}
	
	\begin{example}
		We keep the same parameters $n=10,\ell=4,t=3$ and the same adapter sequences $\mathbf{s}_0=\mathtt{ATGC},\mathbf{s}_{14}=\mathtt{TTGC}$ as Example~\ref{ex:prop1}. Consider instead the input string $\tilde{\mathbf{x}}=\mathtt{ATGC\ GCTATAGTAT\ TTGC}$ and the local input $\ell$-mers
		\[
		\mathbf{s}_4=\mathtt{GCTA},\quad
		\mathbf{s}_5=\mathtt{CTAT},\quad
		\mathbf{s}_6=\mathtt{TATA},\quad
		\mathbf{s}_7=\mathtt{ATAG},\quad
		\mathbf{s}_8=\mathtt{TAGT}.
		\]
		Suppose an edit burst replace 
		\[
		(\mathbf{s}_5,\mathbf{s}_6,\mathbf{s}_7)
		=
		(\mathtt{CTAT},\mathtt{TATA},\mathtt{ATAG})
		\]
		by the single \(\ell\)-mer
		\[
		\mathbf{e}_1=\mathtt{CTAG}.
		\]
		For example, this can be realized by deleting $\mathtt{CTAT}$ and $\mathtt{TATA}$, and substituting $\mathtt{ATAG}$ with $\mathtt{CTAG}$. The local output $\ell$-mers are then
		\[
		(\mathbf{s}_4,\mathbf{e}_1,\mathbf{s}_8)
		=
		(\mathtt{GCTA},\mathtt{CTAG},\mathtt{TAGT}),
		\] which is consistent. In the notation of Lemma~\ref{le:consistent-edit}, we have $\tau=3$ and $\rho=1$, so $\mu=\min\{\tau, \rho\}+1=2$ and $\delta=|\tau-\rho|=2$. 
		Since $\tau>\rho$, the lemma asserts the existence of a periodic pattern of length $\delta+\ell-\mu=2+4-2=4$ and period $\delta=2$ in the input string. 
		Indeed,  the local input string contains a periodic pattern $\mathtt{TATA}$ of length $4$ and period $2$. Moreover, note that $\delta \le \ell-\mu+2 = 4$ and the least period of $\mathtt{TATA}$ is 2, which divides $\delta=2$.
	\end{example}
	
	In particular, for $t\leq \ell-1$, every nontrivial consistency-preserving burst with $\tau\neq\rho$ satisfies the conditions in Part~\ref{ite:nontrivial-edit} of Lemma~\ref{le:consistent-edit}, and the divisibility condition is automatic whenever $\max\{\tau,\rho\}\leq t$. Indeed, $\mu+\delta=\max\{\tau,\rho\}+1\leq t+1\leq \ell$, so $\delta\leq \ell-\mu$.
	
	\section{Deletion-correcting codes}
	
	In this section, we consider the deletion-only nanopore channel, i.e., the adversarial edit channel is allowed to delete $\ell$-mers but not insert or substitute them, and assume that the number of deletions satisfies $t\leq \ell-1$.
	Recall from Proposition~\ref{le:consistent} that, by using the smallest gaps between adjacent output $\ell$-mers, one can associate $\mathbf{z}_0^N$ with the output string $\tilde{\mathbf{y}}=\psi(\mathbf{z}_0^N)$. Although the interpolated $\ell$-mers used in this procedure need not be the	actual deleted $\ell$-mers, the resulting string
	$\tilde{\mathbf{y}}$ is still a subsequence of the input string $\tilde{\mathbf{x}}$. Moreover, as the adapter $\ell$-mers are anchored, $\tilde{\mathbf{y}}$ begins with $\mathbf{s}_0$ and ends with $\mathbf{s}_{n+\ell}$. Therefore, removing the adapter $\ell$-mers from $\tilde{\mathbf{y}}$ gives a subsequence of $\mathbf{x}$ of length at least $n-t$, since there are at most $t$ deleted $\ell$-mers. This leads to the following useful reduction from the $q$-ary deletion channel (i.e., a channel that takes strings over $\Sigma_q$ as input and deletes at most $t$ symbols) to the deletion-only nanopore channel. 
	
	\begin{proposition}\label{prop:reduction}
		Every $t$-deletion-correcting code $\cC\subset \Sigma_q^n$ for the $q$-ary deletion channel is also a $t$-deletion-correcting code for the deletion-only nanopore channel.
	\end{proposition}
	
	\begin{proof}
		Let $\mathbf{x}\in\cC$, and let \(\mathbf{z}_0^N\) be the output of the deletion-only nanopore channel with input
		$\tilde{\mathbf{x}}=(\mathbf{s}_0,\mathbf{x},\mathbf{s}_{n+\ell})$, where at most $t$ $\ell$-mers are deleted. Let $\tilde{\mathbf{y}}=\psi(\mathbf{z}_0^N)$.	As discussed above, $\tilde{\mathbf{y}}$ is a subsequence of
		$\tilde{\mathbf{x}}$ obtained by deleting at most $t$ symbols.	After removing the two known adapters, we obtain a subsequence of $\mathbf{x}$ obtained by at most $t$ deletions. Since $\cC$ corrects $t$ symbol deletions, $\mathbf{x}$ is uniquely recoverable.
	\end{proof}
	
	We note that the obstruction to unique recovery in the deletion-only nanopore channel is caused by periodicity. Indeed, under the assumption $t\leq\ell-1$, the true gap between any two consecutive output $\ell$-mers is at most $\ell$. If the true gap between $\mathbf{z}_i$ and $\mathbf{z}_{i+1}$ is $L_i>\ell_i^*$, then the corresponding overlap of length $\ell-\ell_i^*$ between $\mathbf{z}_i$ and $\mathbf{z}_{i+1}$ has period $L_i-\ell_i^*>0$.
	Thus, the discrepancy between the true gap and the smallest	gap can occur only inside a periodic pattern. The following lemma makes this	precise.
	
	\begin{lemma}
		\label{le:ambiguous-gap}
		Assume $t\leq \ell-1$ and that $L_i$ is the true gap between $\mathbf{z}_i$ and $\mathbf{z}_{j+1}$. If $L_i>\ell_i^*$, then a substring of $L_i-\ell_i^*$ symbols of $\tilde{\mathbf{x}}$ is deleted in $\tilde{\mathbf{y}}$, and it is contained in a periodic pattern of length $L_i+\ell-2\ell_i^*$ and period $L_i-\ell_i^*$ in $\tilde{\mathbf{x}}$. Moreover, the least period of this substring divides $L_i-\ell_i^*$. In particular, this deletion removes an integer number of least periods.
	\end{lemma}
	
	\begin{proof}
		Assume $\mathbf{z}_{i}=\mathbf{s}_j$ and $\mathbf{z}_{i+1}=\mathbf{s}_{j+L_i}$. Since the the smallest gap is used to obtain the output string $\tilde{\mathbf{y}}$, the substring $(x_{j+1},\ldots,x_{j+L_i-\ell_i^*})$ is deleted from the input string $\tilde{\mathbf{x}}$. Since both $\ell_i^*$ and $L_i$ are gaps between $\mathbf{z}_i$ and $\mathbf{z}_{i+1}$, we have
		\[
		x_a=x_{a+(L_i-\ell_i^*)},\qquad a=j-(\ell-\ell_i^*)+1,\ldots,j.
		\]
		Clearly, $(x_{j-\ell+\ell_i^*+1},\ldots,x_{j+L_i-\ell_i^*})$ has length $L_i-2\ell_i^*+\ell$ and period $L_i-\ell_i^*$, and has $(x_{j+1},\ldots,x_{j+L_i-\ell_i^*})$ as a substring. Since $L_i\leq t+1\leq \ell$, we have $L_i-2\ell_i^*+\ell\geq 2(L_i-\ell_i^*)$. Thus, by Remark~\ref{re:FW2}, its least period divides $L_i-\ell_i^*$. It follows the deletion removes an integer number of least periods.
	\end{proof}

	\begin{example}
		We keep the same parameters $n=10,\ell=4,t=3$ and the same adapter sequences $\mathbf{s}_0=\mathtt{ATGC},\mathbf{s}_{14}=\mathtt{TTGC}$ as Example~\ref{ex:prop1}. Consider the input string $\tilde{\mathbf{x}}=\mathtt{ATGC\ ACTTTTTGTC\ TTGC}$ and the local input $\ell$-mers\[
		\mathbf{s}_4=\mathtt{ACTT},\quad
		\mathbf{s}_5=\mathtt{CTTT},\quad
		\mathbf{s}_6=\mathtt{TTTT},\quad
		\mathbf{s}_7=\mathtt{TTTT},\quad
		\mathbf{s}_8=\mathtt{TTTG}.
		\]
		Suppose the deletion channel deletes
		\[
		(\mathbf{s}_5,\mathbf{s}_6,\mathbf{s}_7)
		=
		(\mathtt{CTTT},\mathtt{TTTT},\mathtt{TTTT}).
		\] So locally we have
		\[
		\mathbf{z}_4=\mathbf{s}_4=\mathtt{ACTT},
		\qquad
		\mathbf{z}_5=\mathbf{s}_8=\mathtt{TTTG}.
		\]
		The true gap between $\mathbf{z}_4$ and $\mathbf{z}_5$ is $L_4=8-4=4$, whereas the smallest gap is $\ell_4^*=2$. The length-$\ell^*_4$ suffix of $\mathbf{z}_5$ is $\mathbf{t}_5=\mathtt{TG}$, so the local output string is $(\mathbf{z}_4, \mathbf{t}_5)=\mathtt{ACTTTG}$.
		Equivalently, compared with the local input string, the consistency recovery (i.e., forming $(\mathbf{z}_4, \mathbf{t}_5)$) effectively deletes $L_4-\ell_4^*=2$ symbols. Moreover, the output string is given by $\tilde{\mathbf{y}}=\mathtt{ATGC\ ACTTTGTC\ TTGC}$. 
		Lemma~\ref{le:ambiguous-gap} asserts that the effective deletion is contained in a periodic pattern of length
		$L_4+\ell-2\ell_4^*=4+4-2\cdot 2=4$ and period $L_4-\ell_4^*=2$ in $\tilde{\mathbf{x}}$.
		Indeed, the relevant substring is $\mathtt{TTTT}$,	whose least period is $1$, which divides $2$. Thus, the effective deletion removes an integer number of least periods.
	\end{example}
	
	There is a second way to view the obstruction to unique recovery. Setting $\rho=0$ in Lemma~\ref{le:consistent-edit}, we have the following lemma that identifies the periodic structure that allows a deletion burst to preserve consistency.
	
	\begin{lemma}
		\label{le:consistent-del}
		Let $\tau \geq 1$ be an integer. 
		Let a single deletion burst remove $\mathbf{s}_{i+1},\ldots,\mathbf{s}_{i+\tau}$ while leaving the two boundary $\ell$-mers $\mathbf{s}_i$ and $\mathbf{s}_{i+\tau+1}$ unchanged. 
		If $\mathbf{s}_i$ and $\mathbf{s}_{i+\tau+1}$ are consistent, then
		the input string contains a periodic pattern of length $\tau+\ell-1$ with period $\tau$. 
		Moreover, if $\tau\leq \ell+1$, then the least period of this periodic pattern divides $\tau$.
	\end{lemma}
	
	\begin{remark}\label{re:cp}
		Lemma~\ref{le:consistent-del} shows that a consistency-preserving deletion burst of length $\tau$ in the output $\ell$-mers induces a deletion burst of the same length $\tau$ within a periodic pattern of the output string. If this pattern has least period $T$ and $\tau\leq \ell+1$, then it is necessary that $T|\tau$, so a consistency-preserving deletion burst of length $\tau$ removes exactly $\tau/T$ complete periods of the pattern. Conversely, an input string that contains a periodic pattern of length $\tau+\ell-1$ with least period $T$ can admit $\ell$-mer deletions that correspond to at most $\lfloor \tau/T\rfloor$ complete least periods in that pattern, while keeping the output $\ell$-mers consistent.
	\end{remark}

	\subsection{Existential bounds on the optimal code size}\label{sec:bounds}
	
	By Proposition~\ref{prop:reduction}, any $t$-deletion-correcting code for the $q$-ary deletion channel can also be used for the deletion-only nanopore channel. 
	Thus, any lower bound on the maximum size of codes for the $q$-ary deletion channel is also a lower bound for the deletion-only nanopore channel. To the best of our knowledge, the best such bound (for the parameter regime considered in this work) is due to Alon \etal \cite{alon2023logarithmically}, yielding the following lower bound on the maximum size of codes for the deletion-only nanopore channel. 
	
	\begin{theorem}[{\cite[Theorem~1]{alon2023logarithmically}}]\label{thm:lb}
		A $t$-deletion-correcting code $\cC\subset\Sigma_q^n$ for the nanopore channel has size $|\cC|=\Omega(q^nn^{-2t}\log n)$.
	\end{theorem}
	
	\begin{remark}
		By Theorem~\ref{thm:lb}, there exist $t$-deletion-correcting codes of length $n$ over $\Sigma_q$ for the nanopore channel with redundancy at most $2t\log_qn-\log_q\log_2 n+O(1)$. However, this result is nonconstructive. In Section~\ref{sec:con}, we give an explicit construction with redundancy \(2t\log_q n+\Theta(\log\log n)\), which matches the bound in Theorem~\ref{thm:lb} to first order.
	\end{remark}
	
	The reduction above goes only in one direction: the deletion-only nanopore channel is weaker than the $q$-ary deletion channel. Hence, upper bounds on the maximum size of classical $q$-ary deletion-correcting codes do not automatically imply	upper bounds for the nanopore channel. We next derive an upper bound on the size of $t$-deletion-correcting codes for deletion-only the nanopore channel. 
	
	As mentioned earlier, ambiguity for unique recovery arises from periodic patterns. For the upper bound, it suffices to consider the simplest periodic patterns. A string with period one is called \emph{constant}. We consider a relaxed setting in which an adversary in effect deletes symbols only from maximal constant substrings of length at least $\ell$ of the input. Such deletions are realizable by the deletion-only nanopore channel: deleting one $\ell$-mer inside a sufficiently long constant region preserves the local $\ell$-mer consistency and corresponds to deleting one symbol from the input string. 
	The largest possible size of a $t$-deletion-correcting code in this relaxed setting then serves as an upper bound on the maximum code size in the original setting. Intuitively, a code for this relaxed setting cannot contain many codewords that have a large number of long maximal constant substrings. 
	Motivated by this observation, we adapt the ideas of Levenshtein \cite{levenshtein1966binary} and use the probabilistic method. Specifically, we view each symbol of the string to be sequenced as i.i.d.\ and uniformly distributed over $\Sigma_q$, and estimate the number of strings with a small number of long maximal constant substrings from this perspective. Furthermore, since the symbols within a maximal constant substring are not independent, we employ Janson's concentration inequality \cite{janson2004large} for dependent random variables to bound the number of such strings. Below we state Janson's concentration inequality and then present our upper bound.
	
	\begin{theorem}[{Janson's concentration inequality \cite[Theorem~2.1]{janson2004large}}]\label{thm:janson}
		Let $X=\sum_{\alpha\in A}Y_{\alpha}$ where $Y_{\alpha}$ is a random variable with $a_{\alpha}\leq Y_{\alpha}\leq b_{\alpha}$ and some real numbers $a_{\alpha}$ and $b_{\alpha}$. Then for $\eta>0$,
		\begin{align*}
			\Pr(X\geq \mathbb{E}X+\eta)\leq \exp\Big(-2\frac{\eta^2}{\chi^*(\Gamma)\sum_{\alpha\in A}(b_{\alpha}-a_{\alpha})^2}\Big),
		\end{align*}
		where $\Gamma$ is a dependency graph for $\{Y_{\alpha}\}_{\alpha\in A}$ and $\chi^*(\Gamma)$ is the fractional chromatic number of $\Gamma$. The same estimate holds for $\Pr(X\leq \mathbb{E}X-\eta)$.
	\end{theorem}
	
	\begin{theorem}\label{thm:ub}
		A $t$-deletion-correcting code $\cC\subset\Sigma_q^n$ for the nanopore channel has size $|\cC|=O(q^n n^{-t})$.
	\end{theorem}
	\begin{proof}
		Consider a weaker version of the adversarial nanopore channel where deletions only occur within maximal constant substrings of length at least \(\ell\) in $\mathbf{x}$. Denote by $M$ the largest possible size of a $t$-deletion-correcting codes for this weaker channel. Clearly, $|\cC|\leq M$. In the following, we derive an upper bound on $M$. Our derivation is inspired by \cite[Lemma~2]{levenshtein1966binary}. 
		
		Let $\cC^*\subset\Sigma_q^n$ be a $t$-deletion-correcting code for the weaker channel with $|\cC^*|=M$. Partition $\cC^*$ into two subsets $\cC_1,\cC_2$ according to the number of maximal constant substrings in each of the codewords. More precisely, for an string \(\mathbf{x}\in\Sigma_q^n\), let $\norm{\mathbf{x}}$ be the number of maximal constant substrings of length at least $\ell$ in \(\mathbf{x}\). For some integer \(\kappa \geq t\), define
		\begin{align*}
			\cC_1=\{\mathbf{x}\in\cC^*:\norm{\mathbf{x}} > \kappa\},\quad 
			\cC_2=\{\mathbf{x}\in\cC^*:\norm{\mathbf{x}} \leq \kappa\}.
		\end{align*}
		
		Let \(\cS_t(\mathbf{x})\) be the set of all strings of length \(n-t\) that can be obtained from \(\mathbf{x}\) by deleting symbols in maximal constant substrings in \(\mathbf{x}\). Since $\cS_t(\mathbf{x}),\mathbf{x}\in\cC_1$ are disjoint, we have \(\sum_{\mathbf{x}\in\cC_1} \left|\cS_t(\mathbf{x})\right| \leq q^{n-t}\). Moreover, since \(\kappa \geq t\), \(\left|\cS_t(\mathbf{x})\right| \geq \binom{\kappa}{t}>0\). Therefore, we have \(|\cC_1| \leq q^{n-t}/\binom{\kappa}{t}\).
		In the following, we set \(\kappa = \mu-\sqrt{(\ell tn\ln n)/2}\) for sufficiently large $n$ (so that $\kappa\geq t$), where
		\begin{align}
			\mu:=\frac{1+(n-\ell+1)(q-1)}{q^{\ell}}.\label{eq:mu}
		\end{align} 
		As a consequence, we have \(|\cC_1| \leq q^{n-t}/\binom{\kappa}{t} = \Theta(q^{n}n^{-t})\). 
		
		We next bound the size of $\cC_2$. If $\cC_2$ also has size at most $\Theta(q^nn^{-t})$, then $|\cC^*|\leq \Theta(q^nn^{-t})$. Let \(\cD := \{\mathbf{x}\in\Sigma_q^n: \norm{\mathbf{x}} \leq \kappa\}\). Since $\cC_2\subset \cD$, we have \(|\cC_2| \leq |\cD|\).
		We next use a probabilistic approach to show that $|\cD|\leq \Theta(q^nn^{-t})$. 
		Write $\mathbf{x}=(x_1,\ldots,x_n)$. Assume that $x_1,\ldots,x_n$ are i.i.d.\ random variables, each uniformly distributed on \(\Sigma_q\). 
		Then $|\cD|=q^n\Pr(\mathbf{x}\in\cD)$.
		For $\mathbf{x}\in\Sigma_q^n$, define a sequence of indicator random variables \((I_i)_{1\leq i\leq n-\ell+1}\) such that $I_i=1$ if and only if a maximal constant substring of length at least $\ell$ starts at the $i$-th position of $\mathbf{x}$.
		Then we have \(\Pr(\mathbf{x}\in\cD)=\Pr(\norm{\mathbf{x}} \leq \kappa) = \Pr\big(\sum_{i=1}^{n-\ell+1} I_i \leq \kappa\big)\).
		
		In the following, we bound $\Pr\big(\sum_{i=1}^{n-\ell+1} I_i \leq \kappa\big)$ from above using Janson's concentration inequality \cite{janson2004large}.
		To this end, let us compute the expectation of $\sum_{i=1}^{n-\ell+1}I_i$. Note that $$\Pr(I_1=1) = \Pr(x_{1+j} = x_1  \text{ for }  j=1, \dots, \ell-1) = \prod_{i=2}^\ell\Pr(x_i = x_1) = {q^{-\ell+1}}.$$ 
		Moreover, for $2\leq i \leq n-\ell+1 $ we have \[\Pr(I_i=1) = \Pr(x_{i+j} = x_i \text{ for } j=1, \dots, \ell-1 \text{ and } x_i \neq x_{i-1}) = \Pr(x_i \neq x_{i-1})\prod_{i=2}^\ell\Pr(x_i = x_1) = \frac{(q-1)}{{q^{\ell}}}.\] 
		Therefore, $\mathbb{E}\big[\sum_{i=1}^{n-\ell+1}I_i\big] = \sum_{i=1}^{n-\ell+1}\mathbb{E}[I_i] = \mu$, where the last equality follows by \eqref{eq:mu}.
		
		Note that \(I_i\) and \(I_j\) are independent if and only if \(|i-j|\geq\ell\). Let $\cG=(V,E)$ be the interval graph on $n-\ell+1$ vertices where vertex $v_i\in V$ corresponds to closed interval \([i, i+\ell-1]\) and $(v_i, v_j)\in E$ if and only if the corresponding intervals have a non-empty intersection. 
		Therefore, $\cG$ is perfect and its chromatic number equals its clique number, i.e., $\chi(\cG)=\omega(\cG)$ (e.g., see \cite{west2001introduction}).
		Since any \(\ell\) consecutive vertices $v_i, \dots, v_{i+\ell-1}$ form a clique while any set of more than \(\ell\) vertices must include two vertices that are independent and not connected by an edge, we have \(\omega(\cG) = \ell\).
		Note that the fractional chromatic number $\chi^*$ of a graph is bounded between its clique number $\omega$ and chromatic number $\chi$. It follows that $\omega(\cG)\leq \chi^*(\cG)\leq \chi(\cG)=\omega(\cG)$, and therefore, $\chi^*(\cG)=\ell$. 
		Since $\cG$ is a dependency graph for $(I_i)_{1\leq i\leq n-\ell+1}$, by Theorem~\ref{thm:janson} we have
		\begin{align*}
			\Pr\Big(\sum_{i=1}^{n-\ell+1}I_i \leq \kappa \Big) &=\Pr\Big(\sum_{i=1}^{n-\ell+1}I_i \leq \mu -\sqrt{(\ell tn\ln n)/2}\Big)\\
			& = \Pr\Big(\sum_{i=1}^{n-\ell+1}I_i \leq \mathbb{E}\Big[\sum_{i=1}^{n-\ell+1}I_i\Big] - \sqrt{(\ell tn\ln n)/2}\Big)\\
			& \leq \exp\Big(-\frac{tn\ln n}{n-\ell+1}\Big)\\
			& \leq \exp(-t\ln n).
		\end{align*}
		Hence, \(|\cD| = q^{n}\Pr\big(\sum_{i=1}^{n-\ell+1}I_i \leq \kappa\big)\leq \Theta(q^nn^{-t})\). It follows that \(|\cC_2| \leq \Theta(q^nn^{-t})\).
	\end{proof}
	\begin{remark}\label{re:r-lb}
		By Theorem~\ref{thm:ub}, the redundancy of $t$-deletion-correcting codes for the nanopore channel is at least $t\log_q n+\Omega(1)$.
	\end{remark}

	\subsection{Explicit construction based on GRS codes}\label{sec:con}

	In this subsection, we present an explicit construction of deletion-correcting codes for the deletion-only nanopore channel. Throughout this subsection, we assume $\ell-t\geq \max\{t+1,2t-2\}$, i.e., $t\leq \min\{(\ell-1)/2,(\ell+2)/3\}$.
	Under this condition, we construct codes $\cC\subset\Sigma_q^n$ that can correct up to $t$ deletions of $\ell$-mers for the nanopore channel. Specifically, for any input string $\tilde{\mathbf{x}}=(\mathbf{s}_0,\mathbf{x},\mathbf{s}_{n+\ell})$ with $\mathbf{x}\in\cC$, we show that one can recover $\mathbf{x}$ from the output $\ell$-mers $\mathbf{z}_0^{N}$, provided that at most $t$ deletions occur. Before delving into the technical details, below we outline the high-level idea of our construction.
		
	Our idea is based on the observation that the ambiguity in the deletion-only nanopore channel is caused by deletion of periods in periodic patterns. If at least one complete least period in a periodic pattern survives the deletions in the output string, and the length of the pattern in the input string is known \emph{a priori}, then the full pattern can be recovered by extending it using its periodicity. 
	
	As mentioned in Remark~\ref{re:cp}, for a periodic pattern of length $\tau+\ell-1$ with least period $T$ where $\tau\leq t$, it can undergo deletion of at most $\lfloor \tau/T\rfloor$ complete least periods, leaving a substring of length $\tau+\ell-1-\lfloor \tau/T\rfloor T\geq \ell-1$. Thus, at least $\ell-1$ symbols of the periodic pattern survive the deletions. Moreover, under a deletion budget $t<\ell-1$, the true gap between between $\mathbf{z}_i$ and $\mathbf{z}_{i+1}$ satisfies $L_i\leq t+1$. If the true gap is not the same as the smallest gap, i.e., $L_i\geq \ell_i^*+1$, then the overlap created by the smallest gap have length $\ell-\ell_i^*\geq \ell-L_i+1\geq \ell-t$. Therefore, in the worst case, the length-$(\ell-t)$ overlap between the two $\ell$-mers appearing in the output string $\tilde{\mathbf{y}}$ is a substring of a periodic pattern in the input string with least period $T\leq \ell-t$ (cf. Lemma~\ref{le:ambiguous-gap}). Thus, the shortest periodic pattern that must be detected in the output string can have length only $\ell-t$. 
	This motivates us to keep track of the maximal periodic patterns of length at least $\ell-t$ in both the input string $\tilde{\mathbf{x}}$ and the output string $\tilde{\mathbf{y}}$, which we call \emph{check patterns}.
	
	\begin{definition}[Check patterns]\label{def:check}
		A substring of $\mathbf{s}=(s_1,\ldots,s_n)\in\Sigma_q^n$ is a check pattern if it is a maximal periodic pattern with period at most $t$ and length at least \(\ell-t\).
	\end{definition}

	The condition $\ell-t \geq t+1$ ensures that a surviving substring of period at most $t$ is long enough to be detected as periodic, while the condition $\ell-t \geq 2t-2$ ensures that distinct short-period patterns do not overlap too much, as we will see in Lemma~\ref{le:overlap}.
	
	Now if the set of check patterns in the input string $\tilde{\mathbf{x}}$ is the same as in the output string $\tilde{\mathbf{y}}$, up to a reduction in the length of each pattern, then one can establish a one-to-one correspondence between check patterns in $\tilde{\mathbf{x}}$ and $\tilde{\mathbf{y}}$ that matches each check pattern in $\tilde{\mathbf{x}}$ with the same pattern (possibly with reduced length) in $\tilde{\mathbf{y}}$. Thus, if the lengths of the check patterns in $\tilde{\mathbf{x}}$ are known, one can restore the check patterns in $\tilde{\mathbf{y}}$ to their original lengths using their periodicity. We now show that the check patterns present in the input string are preserved in the output string, possibly with reduced length. This follows from the technical lemma below.
	
	\begin{lemma}\label{le:overlap}
		Assume $t\leq \min\{(\ell-1)/2,(\ell+2)/3\}$. Any two adjacent check patterns in a string can overlap in at most $\ell-t-1$ positions.
	\end{lemma}
	
	\begin{proof}
		Since $t\leq (\ell+2)/3$, i.e, $\ell-t\geq 2t-2$, by Theorem~\ref{thm:FW} two adjacent check patterns with different period $T_1,T_2$ can overlap in at most $T_1+T_2-\gcd(T_1,T_2)-1\leq 2t-3 \leq \ell - t -1$ positions. Moreover, if $T_1 = T_2$, then by Theorem~\ref{thm:FW} they can overlap in at most $T_1-1\leq t-1 \leq  \ell -t-1$ positions, where the last inequality follows from $t\leq (\ell-1)/2$. Hence, any two adjacent check patterns in a string can overlap in at most $\ell-t-1$ positions.
	\end{proof}
	
	\begin{lemma}\label{le:one-to-one}
		Assume $t\leq \min\{(\ell-1)/2,(\ell+2)/3\}$. The number of check patterns in $\tilde{\mathbf{x}}$ is the same as in $\tilde{\mathbf{y}}$. In particular, every check pattern in $\tilde{\mathbf{y}}$ arises from exactly one check pattern in $\tilde{\mathbf{x}}$, possibly with reduced length. Moreover, the check patterns in $\tilde{\mathbf{x}}$ start at distinct positions, and the same holds for $\tilde{\mathbf{y}}$.
	\end{lemma} 
	
	\begin{proof}
		As mentioned above, deletion can only remove complete least periods of a periodic pattern in $\tilde{\mathbf{x}}$, leaving a periodic pattern of length at least $\ell-t$. Therefore, each check pattern in $\tilde{\mathbf{y}}$ must come from a check pattern in $\tilde{\mathbf{x}}$, and the number of check pattern in $\tilde{\mathbf{y}}$ is no larger than that in $\tilde{\mathbf{x}}$.
		
		It remains to show that no two check patterns in $\tilde{\mathbf{x}}$ can be combined into a single check pattern in $\tilde{\mathbf{y}}$ through deletion. 
		Note that by Lemma~\ref{le:overlap}, any two adjacent check patterns in a string can overlap in at most $\ell-t-1$ positions. 
		Since a check pattern has length at least $\ell -t$, this implies they cannot start (or end) at the same position, i.e., one pattern cannot be contained in another. For the same reason, this also implies any two adjacent check patterns in $\tilde{\mathbf{x}}$ remain to be distinct in $\tilde{\mathbf{y}}$. So we are left to show that two non-adjacent check patterns cannot be combined by deletion. For two non-adjacent check patterns in $\tilde{\mathbf{x}}$ to form a single check pattern in $\tilde{\mathbf{y}}$, the substring of $\tilde{\mathbf{x}}$ that separates them must be deleted. Moreover, this separating substring must lie within another check pattern in $\tilde{\mathbf{x}}$, which is adjacent to the two non-adjacent check patterns. By the argument above, this adjacent check pattern cannot be merged with either of the two, and thus any two non-adjacent check patterns remain distinct in $\tilde{\mathbf{y}}$.
	\end{proof}
	
	By Lemma~\ref{le:one-to-one}, we can identify each check pattern in $\tilde{\mathbf{x}}$ with a unique check pattern in  $\tilde{\mathbf{y}}$, sequentially from left to right. If the information of the lengths of the check patterns in $\tilde{\mathbf{x}}$ can be embedded in $\tilde{\mathbf{x}}$ and preserved in $\tilde{\mathbf{y}}$, then recovering $\tilde{\mathbf{x}}$ from $\tilde{\mathbf{y}}$ should be straightforward. 
	The remaining task, therefore, is to design a method to encode these lengths in a way that is robust to deletions. To this end, we first examine how many check patterns in $\tilde{\mathbf{y}}$ may have reduced length.
	
	\begin{lemma}\label{le:error}
		Assume $t\leq \min\{(\ell-1)/2,(\ell+2)/3\}$ and that the number of deleted $\ell$-mer is $\tau\leq t$. At most $\tau$ check patterns in $\tilde{\mathbf{y}}$ can have length smaller than their corresponding check patterns in $\tilde{\mathbf{x}}$.
	\end{lemma}
	
	\begin{proof}
		First note that applying the mapping $\psi$ to $\mathbf{z}_0^{N}$ can only shorten the effective symbol-deletion bursts in $\tilde{\mathbf{y}}$. Indeed, let $L_i$ be the actual gap between $\mathbf{z}_i$ and $\mathbf{z}_{i+1}$. Then $L_i\geq \ell_i^*$. Moreover, it follows that the original $\ell$-mer deletion burst of length $L_i-1$ induces an effective symbol-deletion burst of length $L_i-\ell_i^*\leq L_i-1$.
		Consequently, after applying $\psi$, the output string $\tilde{\mathbf{y}}$ may be viewed as being obtained from $\tilde{\mathbf{x}}$ by disjoint effective deletion bursts whose total length,	and hence whose number, is at most $\tau$. Therefore, for the rest of the proof, we may assume that the output is consistent and work with these effective deletion bursts.
		
		By Lemma~\ref{le:one-to-one}, there is a one-to-one correspondence between check patterns in the input string $\tilde{\mathbf{x}}$ and the output string $\tilde{\mathbf{y}}$ that matches each check pattern in $\tilde{\mathbf{x}}$ with the same pattern (possibly with reduced length) in $\tilde{\mathbf{y}}$.
		We next show that a single deletion burst can reduce the length of at most one check pattern in the input string $\tilde{\mathbf{x}}$. 
		
		Consider an effective deletion burst of length $b\leq \tau$ that deletes the symbols $(x_{i+1},\ldots,x_{i+b})$. 
		By Lemma~\ref{le:consistent-del}, the output $\ell$-mers remain consistent if and only if $(x_{i-\ell+2},\dots,x_{i+b})$ is a periodic pattern with least period $T$ satisfying $T\mid b$.
		Let $\mathbf{p}_1$ be the maximal periodic pattern in $\tilde{\mathbf{x}}$ with least period $T$ that contains $(x_{i-\ell+2},\dots,x_{i+b})$. It follows that $\mathbf{p}_1$ is a check pattern.	
		
		Suppose the same burst also shortens the length of another check pattern $\mathbf{p}_2$. Then the burst must overlap $\mathbf{p}_2$, and thus $\mathbf{p}_1$ overlaps $\mathbf{p}_2$. By Lemma~\ref{le:overlap}, the size of the overlap between $\mathbf{p}_1$ and $\mathbf{p}_2$ is at most $\ell-t-1$. 
		Since check patterns have length at least $\ell-t$, it follows that neither $\mathbf{p}_1$ nor $\mathbf{p}_2$ is contained in the other, and thus the overlap between $\mathbf{p}_1$ and $\mathbf{p}_2$ must be a suffix of one and a prefix of the other. 
		Since the length of $\mathbf{p}_1$ is at least $b+\ell-1$, it follows that $\mathbf{p}_1$ contains at least $b+\ell-1-(\ell-t-1) > b$ consecutive symbols that are disjoint from $\mathbf{p}_2$.
		
		Since $\mathbf{p}_1$ has period $T$ and $T\mid b$, deleting any $b$ consecutive symbols in $\mathbf{p}_1$ produces the same resulting string as deleting $(x_{i+1},\ldots,x_{i+b})$. Therefore, one may assume without loss of generality that the burst deletes any $b$ consecutive symbols of $\mathbf{p}_1$. In particular, one may replace the original burst by an equivalent burst of length $b$ that is disjoint from $\mathbf{p}_2$, without changing the output string. Under this equivalent burst, $\mathbf{p}_2$ is unaffected and cannot be shortened by this burst, which is a contradiction.
		
		Therefore, a single deletion burst can shorten at most one check pattern. Since there are at most $\tau$ bursts, at most $\tau$ check patterns in $\tilde{\mathbf{y}}$ can have length smaller than their corresponding check patterns in $\tilde{\mathbf{x}}$.		
	\end{proof}
	
	As a consequence of Lemma~\ref{le:error}, if we record the lengths of the check patterns in $\tilde{\mathbf{x}}$ sequentially in one tuple and those in $\tilde{\mathbf{y}}$ in another, then the Hamming distance between these two tuples is at most $t$. Therefore, a Hamming-type error-correcting code with minimum distance at least $2t+1$ may be used to protect these lengths. 
	To realize this idea, we construct the input string $\tilde{\mathbf{x}}$ in two stages. We first consider its prefix $\bar{\mathbf{x}}:=(\mathbf{s}_0,\mathbf{u})$, where $\mathbf{u}\in\Sigma_q^k$ denotes an information string, and then compute from $\bar{\mathbf{x}}$ the check-pattern information that will be stored in a separate string appended to $\bar{\mathbf{x}}$.
	Note that the discussion above for $\tilde{\mathbf{x}}$ and $\tilde{\mathbf{y}}$ applies equally to $\bar{\mathbf{x}}$ and its corresponding output string $\bar{\mathbf{y}}$ induced by the nanopore channel. 
	By Lemma~\ref{le:one-to-one}, the check patterns all start at distinct positions, and thus the number of check patterns in $\bar{\mathbf{x}}=(\mathbf{s}_0,\mathbf{u})\in\Sigma_q^{k+\ell}$ is at most $k+\ell-(\ell-t)+1=k+t+1$.
	The suffix to be appended to $\bar{\mathbf{x}}$ is obtained by first computing the syndrome of the tuple of check-pattern lengths of $\bar{\mathbf{x}}$ with respect to a parity-check matrix of a GRS code of length $m:=k+t+1$ with distance $2t+1$.
	This syndrome is then further encoded by a $q$-ary $t$-deletion-correcting code from \cite{sima2020optimal2}. Finally, the resulting encoding is appended to $\bar{\mathbf{x}}$, followed by $\mathbf{s}_{n+\ell}$, to form the full input string $\tilde{\mathbf{x}}$. We are now ready to describe our code construction. An illustration of the encoding workflow is shown in Figure~\ref{fig:workflow}.
	\begin{figure}[!t]
		\centering
		\scalebox{0.8}{\begin{tikzpicture}[
				block/.style={draw, rounded corners, thick, fill=white, minimum width=5.2cm, minimum height=0.85cm, align=center, font=\small},
				arrow/.style={->, >=Stealth, shorten >=2pt, shorten <=2pt},
				maplabel/.style={font=\footnotesize, fill=white, inner sep=1pt, midway, right=2pt},
				]
				
				% === Input ===
				\node[block, fill=blue!8] (xbar) at (0,0) {$\bar{\mathbf{x}} = (\mathbf{s}_0,\mathbf{u})\in\Sigma_q^{k+\ell}$};
				
				% === Step 1: record check-pattern lengths ===
				\node[block, fill=orange!8] (phixbar) at (0,-1.8) {$\phi(\bar{\mathbf{x}}) \in \{0,\ell-t,\ell-t+1,\dots,k+\ell\}^m$\\ {\footnotesize ($m$-tuple of check-pattern lengths of $\bar{\mathbf{x}}$)}};
				
				\draw[arrow] (xbar) -- (phixbar) node[maplabel] {$\phi \colon \Sigma_q^{\le k+\ell} \to \{0,\ell-t,\ell-t+1,\ldots,k+\ell\}^m$};
				
				% === Step 2: encode into a check vector ===
				\node[block, fill=orange!15] (Phixbar) at (0,-3.6) {$\Phi(\bar{\mathbf{x}}):=(f(\phi(\bar{\mathbf{x}})_1),\ldots,f(\phi(\bar{\mathbf{x}})_m))\in\mathbb{F}_p^m$\\ {\footnotesize (check vector of $\bar{\mathbf{x}}$)}};
				
				\draw[arrow] (phixbar) -- (Phixbar) node[maplabel] {$f\colon\{0, \ell-t,\ell-t+1,\dots, k+\ell\}\to\mathbb{F}_p$};
				
				% === Step 3: generate syndrome by H ===
				\node[block, fill=purple!10] (syndrome) at (0,-5.4) {$H \Phi(\bar{\mathbf{x}})^\top \in \mathbb{F}_p^{2t}$\\ {\footnotesize (syndrome of the check vector)}};
				
				\draw[arrow] (Phixbar) -- (syndrome) node[maplabel] {Multiply $\Phi(\bar{\mathbf{x}})$ by a parity-check matrix $H$ of an $(m,m-2t)$ GRS code};
				
				% === Step 4: q-ary representation ===
				\node[block, fill=green!10] (v) at (0,-7.2) {$\mathbf{v} \in \Sigma_q^{2t\lceil\log_q p\rceil}$\\ {\footnotesize ($q$-ary representation of the syndrome)}};
				
				\draw[arrow] (syndrome) -- (v) node[maplabel] {$g: \mathbb{F}_p \to \Sigma_q^{\lceil\log_q p\rceil}$};
				
				% === Step 5: t-symbol-deletion-correcting code ===
				\node[block, fill=green!25] (R) at (0,-9.0) {$(\mathbf{u},\cR_t(\mathbf{s}_0,\mathbf{u})) \in \Sigma_q^n$};
				
				\draw[arrow] (v) -- (R) node[maplabel] {Encode $\mathbf{v}$ into $\cR_t(\mathbf{s}_0,\mathbf{u})$ by a $t$-symbol-deletion-correcting code};
			\end{tikzpicture}
		}
		\caption{Encoding workflow of Construction~\ref{con:rs}.}
		\label{fig:workflow}
	\end{figure}
	
	\begin{construction}
		\label{con:rs}
		Let $\phi \colon \Sigma_q^{\le k+\ell} \to \{0,\ell-t,\ell-t+1,\ldots,k+\ell\}^m$ be a function that maps a $q$-ary string $\mathbf{s}$ of length at most $k+\ell$ to the $m$-tuple of its check pattern lengths, listed in order of their starting positions and padded with zeros if $\mathbf{s}$ contains fewer than $m$ check patterns. Note that each element in the $m$-tuple $\phi(\bar{\mathbf{x}})$ can take $k+t+2$ different values. As we would like to compute the syndrome of $\phi(\bar{\mathbf{x}})$ using a parity-check matrix of a length-$m$ GRS code, the field size of the code should be at least $\max\{k+t+2,m\}=k+t+2$. 
		
		Let $p$ be a prime power such that $p\geq k+t+2$. Let $f\colon\{0, \ell-t,\ell-t+1,\dots, k+\ell\}\to\mathbb{F}_p$ be an injective mapping.
		Furthermore, for a string $\mathbf{s}$ of length at most $k+\ell$, write $\phi(\mathbf{s})=(\phi(\mathbf{s})_1,\ldots,\phi(\mathbf{s})_m)$ and
		define $\Phi \colon \Sigma_q^{\le k+\ell} \to\mathbb{F}_p^m$ by $\Phi(\mathbf{s}):=(f(\phi(\mathbf{s})_1),\ldots,f(\phi(\mathbf{s})_m))$. We refer to $\Phi(\mathbf{s})$ as the \emph{check vector} of $\mathbf{s}$.
		
		Let $\mathbf{u}\in\Sigma_q^k$ be a length-$k$ information string over $\Sigma_q$. Let $H$ be a $2t\times m$ parity-check matrix of an $(m,m-2t)$ GRS code over $\mathbb{F}_p$. We now compute the syndrome of the check vector of $\bar{\mathbf{x}}=(\mathbf{s}_0,\mathbf{u})$ with respect to $H$ and represent it as a $q$-ary string over $\Sigma_q$. Note that the syndrome $H\Phi(\bar{\mathbf{x}})^\top$ is a length-$(2t)$ vector over $\mathbb{F}_p$. 
		Fix an injective mapping $g\colon\mathbb{F}_p\to\Sigma_q^{\lceil \log_q p\rceil}$. Let $\mathbf{v}\in \Sigma_q^{2t\lceil \log_q p\rceil}$ be the $q$-ary string obtained by applying $g$ to each entry of $H\Phi(\bar{\mathbf{x}})^\top$ and concatenating the results.
		
		We next protect the $q$-ary string $\mathbf{v}$ against $t$ deletions. To do so, we encode $\mathbf{v}$ using a $q$-ary $t$-deletion-correcting code, for example, one from \cite{sima2020optimal2}. Denote the resulting encoding of $\mathbf{v}$ by $\cR_t(\bar{\mathbf{x}})=\cR_t(\mathbf{s}_0,\mathbf{u})$. Finally, given the adapter sequences $\mathbf{s}_0$ and $\mathbf{s}_{n+\ell}$, our $t$-deletion-correcting code $\cC$ of length $n$ over $\Sigma_q$ for the nanopore channel is defined to be 
		\begin{align*}
			\cC=\{(\mathbf{u},\cR_t(\mathbf{s}_0,\mathbf{u}))\mid \mathbf{u}\in\Sigma_q^k\}.
		\end{align*}
	\end{construction}
	
	\begin{remark}
		Note that the length of the $q$-ary string $\mathbf{v}$ representing the syndrome of the check vector for $\bar{\mathbf{x}}$ is $2t\lceil \log_q p\rceil\leq 2t\lceil \log_q (2k+2t+4)\rceil$, since there always exists at least one prime number between $k+t+2$ and $2(k+t+2)$.
		If $q\leq 2t\lceil \log_q p\rceil$, which is typically the case for DNA sequencing, then by the results of \cite{sima2020optimal2}, the length of \(\cR_t(\bar{\mathbf{x}})\) is \(2t\log_q k+\Theta(\log \log k)\), assuming $k$ grows while $t,q$ are fixed constants independent of $k$. Therefore, the redundancy of the code $\cC$ is \(2t\log_q n+\Theta(\log\log n)\).
	\end{remark}
	
	We next show that any codeword of the code given in Construction~\ref{con:rs} can be recovered from the output $\ell$-mers of the deletion-only nanopore channel, provided that $\tilde{\mathbf{x}}=(\mathbf{s}_0,\mathbf{x},\mathbf{s}_{n+\ell})$ is the input and at most  $t\leq \min\{(\ell-1)/2,(\ell+2)/3\}$ deletions occur in the output. 
	
	\begin{theorem}\label{thm:con1}
		Assume $t\leq \min\{(\ell-1)/2,(\ell+2)/3\}$.
		The code $\cC\subset\Sigma_q^n$ described in Construction~\ref{con:rs} is a $t$-deletion-correcting code for the nanopore channel with redundancy \(2t\log_q n+\Theta(\log\log n)\).
	\end{theorem} 
	
	\begin{proof}
		Let $\tau\leq t$ be the number of effective symbol deletions in the output string $\tilde{\mathbf{y}}$. So the length of $\tilde{\mathbf{y}}$ is $n+2\ell-\tau$. Write \(\tilde{\mathbf{y}}=(y_{-\ell+1}, \dots, y_{n+\ell-\tau})\).
		Denote the length of $\cR_t(\bar{\mathbf{x}})=\cR_t(\mathbf{s}_0,\mathbf{u})$ by $r$, so the length of $\tilde{\mathbf{x}}$ is \(n+2\ell = k+r+2\ell\). 
		
		%	\paragraph*{Step 1}
		The first step of our decoding procedure is to recover the syndrome $H\Phi(\bar{\mathbf{x}})^\top$ over $\mathbb{F}_p$.
		Since $\tau$ symbols of $\tilde{\mathbf{x}}$ are deleted, $\cR_t(\bar{\mathbf{x}})$ can get shifted in $\tilde{\mathbf{y}}$ or deleted by at most $\tau$ positions. Therefore, $(y_{k+1},\dots,y_{k+r-\tau})$ must be a length-$(r-\tau)$ subsequence of $\cR_t(\bar{\mathbf{x}})$. Since $\cR_t(\bar{\mathbf{x}})$ is a codeword in a $q$-ary $t$-deletion-correcting code, we can recover $\cR_t(\bar{\mathbf{x}})$ from $(y_{k+1},\dots,y_{k+r-\tau})$ and further obtain the $q$-ary string $\mathbf{v}$ that represents the syndrome of the check vector $\Phi(\bar{\mathbf{x}})$ with respect to $H$. Applying the mapping $g^{-1}$ to every $\lceil\log_q p\rceil$ symbols of $\mathbf{v}$, we recover the syndrome $H\Phi(\bar{\mathbf{x}})^\top$.

		%	\paragraph*{Step 2}
		We next recover the check-pattern lengths $\phi(\bar{\mathbf{x}})$. Let $\hat{\mathbf{y}} = (y_{-\ell+1},\dots,y_{k-\tau})$ be the length-($k+\ell-\tau$) prefix of $\tilde{\mathbf{y}}$. Since there are $\tau$ deletions, $\hat{\mathbf{y}}$ is a subsequence of $\bar{\mathbf{x}}$. 
		Let $\bar{\mathbf{y}}$ be the substring of the output string $\tilde{\mathbf{y}}$ that corresponds to $\bar{\mathbf{x}}$. Observe that $\hat{\mathbf{y}}=\bar{\mathbf{y}}$ if the $\tau$ symbol deletions induced by the nanopore channel all occur in the first $\ell+k$ positions of the input string, but generally $\hat{\mathbf{y}}$ is a prefix of $\bar{\mathbf{y}}$. The relation between different substrings of $\tilde{\mathbf{x}}$ and $\tilde{\mathbf{y}}$ is illustrated in Figure~\ref{fig:relation}.
		\begin{figure}[!t]
			\centering
			\scalebox{0.8}{\begin{tikzpicture}[x=1cm, y=1cm]
					% --- Vertical dashed line at end of \bar{y}: x = 5.0 ---
					
					% ============================================================
					% RECTANGLE 1: \tilde{\mathbf{x}}
					% ============================================================
					% Outer outline
					\draw[very thick] (0, 4.5) rectangle (12.0, 5.5);
					
					% Partition lines
					\draw[thick] (6.5, 4.5) -- (6.5, 5.5);
					\draw[thick] (11.0, 4.5) -- (11.0, 5.5);
					
					% Fill
					\fill[blue!15]   (0,    4.5) rectangle (6.5,  5.5);
					\fill[green!20]  (6.5,  4.5) rectangle (11.0, 5.5);
					\fill[orange!15] (11.0, 4.5) rectangle (12.0, 5.5);
					
					\draw[dashed, very thick, gray!60] (5.0, 3.5) -- (6.5, 4.5);
					\draw[dashed, very thick, gray!60] (8.5, 3.5) -- (11, 4.5);
					% Labels
					\node[font=\bfseries] at (3.25, 5.0)  {$\bar{\mathbf{x}}$};
					\node[font=\bfseries] at (9.25, 5.0)  {$\cR_t(\mathbf{s}_0,\mathbf{u})$};
					\node[font=\bfseries] at (11.5, 5.0)  {$\mathbf{s}_{n+\ell}$};
					
					% Length braces above
					\draw[decorate, decoration={brace, amplitude=5pt}]
					(0,    5.55) -- (6.5,  5.55) node[midway, above=2pt, font=\footnotesize] {$k+\ell$};
					\draw[decorate, decoration={brace, amplitude=5pt}]
					(6.5,  5.55) -- (11.0, 5.55) node[midway, above=2pt, font=\footnotesize] {$r$};
					\draw[decorate, decoration={brace, amplitude=5pt}]
					(11.0, 5.55) -- (12.0, 5.55) node[midway, above=2pt, font=\footnotesize] {$\ell$};
					
					\draw[decorate, decoration={brace, amplitude=5pt, mirror}]
					(0.0, 4.5) -- (12.0, 4.5) node[midway, below=4pt, font=\footnotesize] {$\tilde{\mathbf{x}}$};
					
					% ============================================================
					% RECTANGLE 2: \tilde{\mathbf{y}}
					% ============================================================
					% Outer outline
					
					\draw[very thick] (0, 2.5) rectangle (9.5, 3.5);
					
					% Partition lines
					\draw[thick] (5.0, 2.5) -- (5.0, 3.5);
					\draw[thick] (8.5, 2.5) -- (8.5, 3.5);
					
					% Fill (same colors)
					\fill[blue!15]   (0,    2.5) rectangle (5.0, 3.5);
					\fill[green!20]  (5.0,  2.5) rectangle (8.5, 3.5);
					\fill[orange!15] (8.5,  2.5) rectangle (9.5, 3.5);
					
					% Labels
					\node[font=\bfseries] at (2.5,  3.0) {$\bar{\mathbf{y}}$};
					\node[font=\bfseries] at (6.75, 3.15) {$\cR_t(\mathbf{s}_0,\mathbf{u})$};
					\node[font=\footnotesize] at (6.75, 2.8) {(possibly corrupted)};
					\node[font=\bfseries] at (9.0,  3.0) {$\mathbf{s}_{n+\ell}$};
					
					% Length braces above
					\draw[decorate, decoration={brace, amplitude=5pt}]
					(0,    3.55) -- (5.0, 3.55) node[midway, above=2pt, font=\footnotesize] {$k+\ell-\tau^*$};
					\draw[decorate, decoration={brace, amplitude=5pt}]
					(5.0,  3.55) -- (8.5, 3.55) node[midway, above=2pt, font=\footnotesize] {$r-(\tau-\tau^*)$};
					\draw[decorate, decoration={brace, amplitude=5pt}]
					(8.5,  3.55) -- (9.5, 3.55) node[midway, above=2pt, font=\footnotesize] {$\ell$};
					
					\draw[decorate, decoration={brace, amplitude=5pt, mirror}]
					(0,  2.5) -- (9.5, 2.5) node[midway, below=4pt, font=\footnotesize] {$\tilde{\mathbf{y}}$};
					
					% ============================================================
					% RECTANGLE 3: \tilde{\mathbf{y}} with \bar{\mathbf{y}} partitioned
					% ============================================================
					% Outer outline
					\draw[very thick] (0, 0.0) rectangle (9.5, 1.0);
					
					% Partition lines for intervals 2 and 3
					\draw[thick] (5.0, 0.0) -- (5.0, 1.0);
					\draw[thick] (8.5, 0.0) -- (8.5, 1.0);
					
					% Sub-partition inside \bar{y}
					\draw[thick] (4.2, 0.0) -- (4.2, 1.0);
					
					% Fill
					\fill[red!18]    (0,    0.0) rectangle (4.2, 1.0);
					\fill[blue!15]   (4.2,  0.0) rectangle (5.0, 1.0);
					\fill[green!20]  (5.0,  0.0) rectangle (8.5, 1.0);
					\fill[orange!15] (8.5,  0.0) rectangle (9.5, 1.0);
					
					% Labels
					\node[font=\bfseries] at (2.1,  0.5)  {$\hat{\mathbf{y}}$};
					\node[font=\bfseries] at (6.75, 0.65)  {$\cR_t(\mathbf{s}_0,\mathbf{u})$};
					\node[font=\footnotesize] at (6.75, 0.3) {(possibly corrupted)};
					\node[font=\bfseries] at (9.0,  0.5)  {$\mathbf{s}_{n+\ell}$};
					
					% Length braces above
					\draw[decorate, decoration={brace, amplitude=5pt}]
					(0,    1.05) -- (4.2, 1.05) node[midway, above=2pt, font=\footnotesize] {$k+\ell-\tau$};
					\draw[decorate, decoration={brace, amplitude=5pt}]
					(4.2,  1.05) -- (5.0, 1.05) node[midway, above=2pt, font=\footnotesize] {$\tau-\tau^*$};
					\draw[decorate, decoration={brace, amplitude=5pt}]
					(5.0,  1.05) -- (8.5, 1.05) node[midway, above=2pt, font=\footnotesize] {$r-(\tau-\tau^*)$};
					\draw[decorate, decoration={brace, amplitude=5pt}]
					(8.5,  1.05) -- (9.5, 1.05) node[midway, above=2pt, font=\footnotesize] {$\ell$};
					
					\draw[decorate, decoration={brace, amplitude=5pt, mirror}]
					(0,  0) -- (9.5, 0) node[midway, below=4pt, font=\footnotesize] {$\tilde{\mathbf{y}}$};
					
					% ============================================================
					% VERTICAL DASHED LINE at x = 5.0 (end of \bar{y})
					% ============================================================
					\draw[dashed, very thick, gray] (5.0, -0.6) -- (5.0, 6.0);
					
				\end{tikzpicture}
				}
    			\caption{Relation between different substrings of the input string $\tilde{\mathbf{x}}$ and the output string $\tilde{\mathbf{y}}$ that are involved in the decoding of Construction~\ref{con:rs}.}
				\label{fig:relation}
		\end{figure}
		Note that Lemma~\ref{le:one-to-one} and \ref{le:error} also hold for $\bar{\mathbf{x}}$ and $\bar{\mathbf{y}}$. Therefore, the Hamming distance between $\phi(\bar{\mathbf{x}})$ and $\phi(\bar{\mathbf{y}})$ is at most $\tau$. 
		Since a check pattern has length at least $\ell-t>t\geq \tau$ and $\hat{\mathbf{y}}$ is a prefix of $\bar{\mathbf{y}}$ whose length is smaller by at most $\tau$, the check patterns in $\hat{\mathbf{y}}$ and $\bar{\mathbf{y}}$ are identical, except possibly for the last one. 
		Specifically, the last check pattern in $\bar{\mathbf{y}}$ may get shortened or destroyed in $\hat{\mathbf{y}}$.
		It then follows that the Hamming distance between $\phi(\bar{\mathbf{x}})$ and $\phi(\hat{\mathbf{y}})$ is at most $\tau$. 
		Noticing that $f$ is an injective mapping, the Hamming distance between $\Phi(\bar{\mathbf{x}})$ and $\Phi(\hat{\mathbf{y}})$ is also at most $\tau$. Recall that $H$ is a parity-check matrix for an $(m,m-2t)$ GRS code over $\mathbb{F}_p$ with minimum distance $2t+1$. Using a variant of the decoders for GRS codes, we can decode $\Phi(\bar{\mathbf{x}})$ from $H\Phi(\hat{\mathbf{y}})^\top-H\Phi(\bar{\mathbf{x}})^\top$, since this quantity is the syndrome of the error vector $\Phi(\hat{\mathbf{y}})-\Phi(\bar{\mathbf{x}})$, which has Hamming weight most $\tau\leq t$. Since $f$ is injective, we can further recover $\phi(\bar{\mathbf{x}})$.
		
		%	\paragraph*{Step 3}
		It remains to show that $\bar{\mathbf{x}}$ can be recovered given $\tilde{\mathbf{y}}$ and $\phi(\bar{\mathbf{x}})$. If $\hat{\mathbf{y}}$ is exactly $\bar{\mathbf{y}}$, then by Lemma~\ref{le:ambiguous-gap}, Remark~\ref{re:cp} and Lemma~\ref{le:one-to-one}, $\bar{\mathbf{x}}$ can be found by simply extending each check pattern in $\hat{\mathbf{y}}$ so that its length matches that of its counterpart in $\bar{\mathbf{x}}$. 
		However, $\hat{\mathbf{y}}$ is generally only a prefix of $\bar{\mathbf{y}}$. So a careful analysis is needed to distinguish different cases.
		
		Define $m_{\bar{\mathbf{x}}}$ to be the number of check patterns in $\bar{\mathbf{x}}$, so we have $0\leq m_{\bar{\mathbf{x}}}\leq m$. Similarly, define $m_{\hat{\mathbf{y}}}$ to be the number of check patterns in $\hat{\mathbf{y}}$.  Moreover, we have $m_{\hat{\mathbf{y}}}\in\{m_{\bar{\mathbf{x}}}-1, m_{\bar{\mathbf{x}}}\}$ since the last check pattern in $\bar{\mathbf{y}}$ (and thus in $\bar{\mathbf{x}}$) may get destroyed in $\hat{\mathbf{y}}$.
		
		Consider the case where $m_{\hat{\mathbf{y}}}=m_{\bar{\mathbf{x}}}$. In this case, the last check pattern in $\bar{\mathbf{x}}$ may remain intact or get shortened in the $\hat{\mathbf{y}}$. Therefore, we can obtain a string $\hat{\mathbf{x}}$ from $\hat{\mathbf{y}}$ and $\phi(\bar{\mathbf{x}})$ by extending the $i$-th check pattern in $\hat{\mathbf{y}}$ to be of length $\phi(\bar{\mathbf{x}})_i$ using its periodicity for all $i=1,\ldots, m_{\hat{\mathbf{y}}}$. Since $\hat{\mathbf{y}}$ is a prefix of $\bar{\mathbf{y}}$ obtained by truncating $\bar{\mathbf{y}}$ to its first $\ell+k-\tau$ symbols, 	
		the length of $\hat{\mathbf{x}}$ is no larger than that of $\bar{\mathbf{x}}$. Denote the length $\hat{\mathbf{x}}$ by $\ell+\hat{k}$. If $\hat{k}=k$, then the truncation occurs right at the ending position of $\bar{\mathbf{y}}$, and thus we have $\hat{\mathbf{x}}=\bar{\mathbf{x}}$.
		
		If $\hat{k}<k$, then some symbols in $\bar{\mathbf{y}}$ that follow its last check pattern are removed when forming $\hat{\mathbf{y}}$. In fact, some symbols of the last check pattern in $\bar{\mathbf{y}}$ may also be removed by the truncation. However, in that case these symbols are already recovered in $\hat{\mathbf{x}}$ by extending the last check pattern in $\hat{\mathbf{y}}$ to have length $\phi(\bar{\mathbf{x}})_{m_{\bar{\mathbf{x}}}}$. At the same time, these symbols still appear in the substring $(y_{k+1-\tau},\ldots,y_{n+\ell-\tau})$ that follows $\hat{\mathbf{y}}$. Thus, we need to distinguish these two cases.

		Let $\tau^*\leq \tau$ be the number of symbol deletions in $\bar{\mathbf{x}}$ induced by the nanopore channel to form $\bar{\mathbf{y}}$. Then the length difference between $\bar{\mathbf{y}}$ and $\hat{\mathbf{y}}$ is $\tau-\tau^*$ and we have $0<k-\hat{k}\leq\tau-\tau^*$. If the last check pattern in $\hat{\mathbf{y}}$ stops before the last symbol in $\hat{\mathbf{y}}$, then no symbols of the last check pattern in $\bar{\mathbf{y}}$ are removed by the truncation. It follows that $k-\hat{k}=\tau-\tau^*$ and we have $(\hat{\mathbf{x}},y_{k+1-\tau},\ldots,y_{2k-\hat{k}-\tau})=\bar{\mathbf{x}}$.
		If the last check pattern in $\hat{\mathbf{y}}$ stops at the last symbol in $\hat{\mathbf{y}}$, then we need to find the prefix of $(y_{k+1-\tau},\ldots,y_{n+\ell-\tau})$ that, when appended to $\hat{\mathbf{y}}$, extends its last check pattern in $\hat{\mathbf{y}}$. Note that such a prefix can be easily found by exploiting the periodicity of the last check pattern in $\hat{\mathbf{y}}$. Let $k+s-\tau $ be the ending index of this prefix where $s\geq 1$. Then $s+k-\hat{k}=\tau-\tau^*$ and thus we have $(\hat{\mathbf{x}},y_{k+s+1-\tau},\ldots,y_{2k+s-\hat{k}-\tau})=\bar{\mathbf{x}}$.	
		
		Lastly, consider the case where $m_{\hat{\mathbf{y}}}=m_{\bar{\mathbf{x}}}-1$. In this case, the last check pattern in $\bar{\mathbf{y}}$ is destroyed by truncating $\bar{\mathbf{y}}$ to $\hat{\mathbf{y}}$. More precisely, the truncation reduces its length to less than $\ell-t$. By definition of $m_{\hat{\mathbf{y}}}$, at most the last $\ell-t-1$ symbols of $\hat{\mathbf{y}}$ (and hence $\hat{\mathbf{x}}$) can belong to this check pattern. To identify it, we search for the first check pattern in $(y_{k-\tau-\ell+t+2},\ldots,y_{n+\ell-\tau})$. We then extend the last $\ell-t-1$ symbols of $\hat{\mathbf{x}}$ up to the first check pattern in $(y_{k-\tau-\ell+t+1},\ldots,y_{n+\ell-\tau})$, and subsequently extend the check pattern to have length $\phi(\bar{\mathbf{x}})_{m_{\bar{\mathbf{x}}}}$ using its periodicity. Note that the resulting string still has length at most $\ell+k$. By the same argument as in the case $m_{\hat{\mathbf{y}}}=m_{\bar{\mathbf{x}}}$, the resulting string either has length exactly $\ell+k$, in which case it coincides with $\bar{\mathbf{x}}$, or it can be further extended to obtain $\bar{\mathbf{x}}$.
		The scenarios considered for recovering $\bar{\mathbf{x}}$ from $\tilde{\mathbf{y}}$ and $\phi(\bar{\mathbf{x}})$ are summarized in Figure~\ref{fig:scenarios}.
		\begin{figure}[!t]
			\centering
			\scalebox{0.6}{
				\begin{tikzpicture}
					% ======= Scenario 1 =======
					\node[anchor=east] at (-0.15, 6.0) {$m_{\hat{\mathbf{y}}}=m_{\bar{\mathbf{x}}}$, $k-\hat{k} = \tau-\tau^*$};
					\draw[very thick] (0, 5.2) rectangle (9.0, 6.8);
					
					% Partition lines
					\draw[thick] (2, 5.2) -- (2, 6.8);
					\draw[thick] (6, 5.2) -- (6, 6.8);
					\draw[thick] (7, 5.2) -- (7, 6.8);
					
					% Fill
					\fill[blue!15]   (7,    5.2) rectangle (9,  6.8);
					\fill[red!18]  (6,  5.2) rectangle (7, 6.8);
					\fill[yellow!30]  (2,  5.2) rectangle (6, 6.8);
					\fill[red!18] (0, 5.2) rectangle (2, 6.8);
					
					% Labels
					\node at (4, 6.0)  {last check pattern of $\bar{\mathbf{y}}$};
					
					% Length braces above
					\draw[decorate, decoration={brace, amplitude=5pt}]
					(7, 6.85) -- (9,  6.85) node[midway, above=3pt, font=\footnotesize] {$\tau-\tau^*$};
					
					\draw[decorate, decoration={brace, amplitude=5pt, mirror}]
					(0.0, 5.15) -- (7, 5.15) node[midway, below=3pt] {$\hat{\mathbf{y}}$};
					
					\draw [->, very thick] (9.5,6) -- (11.5,6) node [midway, above=3pt] {extend};
					
					\draw[very thick] (12, 5.2) rectangle (21.75, 6.8);
					
					% Partition lines
					\draw[thick] (14, 5.2) -- (14, 6.8);
					\draw[thick] (18.75, 5.2) -- (18.75, 6.8);
					\draw[thick] (19.75, 5.2) -- (19.75, 6.8);
					
					% Fill
					\fill[blue!15]   (19.75,    5.2) rectangle (21.75,  6.8);
					\fill[red!18]  (18.75,  5.2) rectangle (19.75, 6.8);
					\fill[orange!50]  (14,  5.2) rectangle (18.75, 6.8);
					\fill[red!18] (12, 5.2) rectangle (14, 6.8);
					
					% Labels
					\node at (16.375, 6.0)  {last check pattern of $\bar{\mathbf{x}}$};
					
					% Length braces above
					\draw[decorate, decoration={brace, amplitude=5pt}]
					(19.75,    6.85) -- (21.75,  6.85) node[midway, above=3pt, font=\footnotesize] {$\tau-\tau^*$};
					
					\draw[decorate, decoration={brace, amplitude=5pt, mirror}]
					(12, 5.15) -- (19.75, 5.15) node[midway, below=3pt] {$\hat{\mathbf{x}}$};
					
					% ======= Scenario 2 =======
					
					%   rect 2 y = 2.6 to 4.2
					%   rect 3 y = 0.0 to 1.6
					\node[anchor=east] at (-0.15, 3.4) {$m_{\hat{\mathbf{y}}}=m_{\bar{\mathbf{x}}}$, $k-\hat{k} < \tau-\tau^*$};
					\draw[very thick] (0, 2.6) rectangle (9, 4.2);
					
					% Partition lines
					\draw[thick] (3.5, 2.6) -- (3.5, 4.2);
					\draw[thick] (7.5, 2.6) -- (7.5, 4.2);
			
					% Fill
					\fill[blue!15]   (7.5, 2.6) rectangle (9,  4.2);
					\fill[yellow!30]  (3.5,  2.6) rectangle (7.5, 4.2);
					\fill[red!18] (0, 2.6) rectangle (3.5, 4.2);
					
					% Labels
					\node at (5.5, 3.4)  {last check pattern of $\bar{\mathbf{y}}$};
					
					% Length braces above
					\draw[decorate, decoration={brace, amplitude=5pt}]
					(3.5,    4.25) -- (7,  4.25) node[midway, above=3pt] {$\ge \ell-t$};
					\draw[decorate, decoration={brace, amplitude=5pt}]
					(7,    4.25) -- (7.5,  4.25) node[midway, above=3pt] {$s$};
					\draw[decorate, decoration={brace, amplitude=5pt}]
					(7.5,    4.25) -- (9,  4.25) node[midway, above=3pt] {$\tau-\tau^*-s$};
					
					\draw[decorate, decoration={brace, amplitude=5pt, mirror}]
					(0.0, 2.55) -- (7, 2.55) node[midway, below=3pt] {$\hat{\mathbf{y}}$};
					
					\draw [->, very thick] (9.5, 3.4) -- (11.5,3.4) node [midway, above=3pt] {search} node [midway, below=3pt] {then extend};
					
					\draw[very thick] (12, 2.6) rectangle (22.25, 4.2);
					
					% Partition lines
					\draw[thick] (15.5, 2.6) -- (15.5, 4.2);
					\draw[thick] (20.25, 2.6) -- (20.25, 4.2);
					\draw[thick] (20.75, 2.6) -- (20.75, 4.2);
					
					% Fill
					\fill[blue!15]   (20.75,    2.6) rectangle (22.25,  4.2);
					\fill[yellow!30]  (20.25,  2.6) rectangle (20.75, 4.2);
					\fill[orange!50]  (15.5,  2.6) rectangle (20.25, 4.2);
					\fill[red!18] (12, 2.6) rectangle (15.5, 4.2);
					
					% Labels
					\node at (17.875, 3.4)  {last check pattern of $\bar{\mathbf{x}}$};
					
					% Length braces above
					\draw[decorate, decoration={brace, amplitude=5pt}]
					(20.25,    4.25) -- (20.75,  4.25) node[midway, above=3pt] {$s$};
					\draw[decorate, decoration={brace, amplitude=5pt}]
					(20.75,    4.25) -- (22.25,  4.25) node[midway, above=3pt] {$\tau-\tau^*-s$};
					
					\draw[decorate, decoration={brace, amplitude=5pt, mirror}]
					(12, 2.55) -- (20.25, 2.55) node[midway, below=3pt] {$\hat{\mathbf{x}}$};
					
					% ======= Scenario 3 =======
					
					%   rect 2 y = 2.6 to 4.2
					%   rect 3 y = 0.0 to 1.6
					\node[anchor=east] at (-0.15, 0.8) {$m_{\hat{\mathbf{y}}}=m_{\bar{\mathbf{x}}}-1$};
					\draw[very thick] (0, 0) rectangle (9, 1.6);
					
					% Partition lines
					\draw[thick] (4, 0) -- (4, 1.6);
					\draw[thick] (8, 0) -- (8, 1.6);
					
					% Fill
					\fill[blue!15]   (8, 0) rectangle (9,  1.6);
					\fill[yellow!30]  (4,  0) rectangle (8, 1.6);
					\fill[red!18] (0, 0) rectangle (4, 1.6);
					
					% Labels
					\node at (6, 0.8)  {last check pattern of $\bar{\mathbf{y}}$};
					
					% Length braces above
					\draw[decorate, decoration={brace, amplitude=5pt}]
					(4,    1.65) -- (7,  1.65) node[midway, above=3pt] {$< \ell-t$};
					\draw[decorate, decoration={brace, amplitude=5pt}]
					(7,    1.65) -- (8,  1.65) node[midway, above=3pt] {$s$};
					\draw[decorate, decoration={brace, amplitude=5pt}]
					(8,    1.65) -- (9,  1.65) node[midway, above=3pt] {$\tau-\tau^*-s$};
					
					\draw[decorate, decoration={brace, amplitude=5pt, mirror}]
					(0.0, -0.05) -- (7, -0.05) node[midway, below=3pt] {$\hat{\mathbf{y}}$};
					
					\draw[dashed, very thick, gray] (7.0, -0.6) -- (7.0, 7.3) node[above, color=black] {the right end of $\hat{\mathbf{y}}$ in $\bar{\mathbf{y}}$};
					
					\draw [->, very thick] (9.5, 0.8) -- (11.5,0.8) node [midway, above=3pt] {search} node [midway, below=3pt] {then extend};
					
					\draw[very thick] (12, 0) rectangle (21.75, 1.6);
					
					% Partition lines
					\draw[thick] (16, 0) -- (16, 1.6);
					\draw[thick] (20.75, 0) -- (20.75, 1.6);
					
					% Fill
					\fill[blue!15]   (20.75,    0) rectangle (21.75,  1.6);
					\fill[orange!50]  (16,  0) rectangle (20.75, 1.6);
					\fill[red!18] (12, 0) rectangle (16, 1.6);
					
					% Labels
					\node at (18.375, 0.8)  {last check pattern of $\bar{\mathbf{x}}$};
					
					% Length braces above
					\draw[decorate, decoration={brace, amplitude=3pt}]
					(20.75,    1.65) -- (21.75,  1.65) node[midway, above=3pt] {$\tau-\tau^*-s$};
					
					\draw[decorate, decoration={brace, amplitude=3pt, mirror}]
					(12, -0.05) -- (20.75, -0.05) node[midway, below=3pt] {$\hat{\mathbf{x}}$};
					\draw[dashed, very thick, gray] (19.0, -0.6) -- (19.0, 7.3) node[above, color=black] {};
				\end{tikzpicture}
			}
			\caption{Possible scenarios for recovering $\bar{\mathbf{x}}$ from $\tilde{\mathbf{y}}$ and $\phi(\bar{\mathbf{x}})$ considered in Theorem~\ref{thm:con1}.}
			\label{fig:scenarios}
		\end{figure}
		After recovering $\bar{\mathbf{x}}$, the string $\mathbf{x}$ can be found by removing the adapter sequence $\mathbf{s}_0$.
	\end{proof}
	
	\begin{remark}
		For constant $\ell$ and $q$, the encoding and decoding procedures above both run in time linear in $n$, provided that the string $\cR(\bar{\mathbf{x}})$ is generated using a $t$-deletion-correcting code from \cite{sima2020optimal2} in the small-$q$ regime. More precisely, computing the check-pattern lengths of $\bar{\mathbf{x}}$ and the syndrome $H\Phi(\bar{\mathbf{x}})^T$ takes $O(n)$ time, whereas applying the encoder of \cite{sima2020optimal2} to form the length-$O(\log n)$ string $\cR(\bar{\mathbf{x}})$ contributes $O((\log n)^{2t})$ time. Likewise, decoding the string $\cR(\bar{\mathbf{x}})$ also takes $O((\log n)^{2t})$ time by \cite{sima2020optimal2}. The remaining decoding  steps, including identifying the check patterns in $\hat{\mathbf{y}}$, decoding the length-$\Theta(n)$ GRS syndrome with a constant number of errors, and reconstructing $\bar{\mathbf{x}}$, each require $O(n)$ time.
	\end{remark}
	
	\begin{remark}
		We stress that the assumption that the adapter sequences are known is essential for decoding.
		If the left adapter $\mathbf{s}_0$ were unknown, then any deleted symbols in the portion of $\mathbf{s}_0$ preceding the leftmost check pattern of the input would be impossible to recover from the channel output, and an analogous ambiguity would occur at the right end if $\mathbf{s}_{n+\ell}$ were not known.
		Moreover, even if the leftmost check pattern starts at the first symbol of $\mathbf{s}_0$, 
		deleting an initial segment whose length is not a multiple of its least period may still yield a consistent output but shift the phase of the periodic pattern, thereby rendering it impossible to determine the pattern's original starting point.
	\end{remark}
	
	\subsection{Explicit construction based on Sidon sets}
	In the GRS-based construction above, the changes in the check-pattern lengths are treated as an arbitrary Hamming error vector of weight at most $t$. However, in the deletion-only setting, the changes in the check-pattern lengths are not arbitrary. In fact, the error vector is a nonnegative integer vector whose entries have a bounded sum. More precisely, if the $j$-th check pattern is shortened by $\tau_j$ symbols, then $\tau_j\geq 0$ and $\sum_j \tau_j\leq t$.
	Thus, a decoder does not need to correct an arbitrary error vector of weight at most $t$. It only needs to identify a multiset of at most $t$ check-pattern indices, where the index $j$ appears $\tau_j$ times. Once the indices of the shortened check patterns are known, together with their multiplicities given by the number of deleted symbols in each pattern, the original string can be recovered by exploiting the periodicity of those patterns.

	This observation leads to more economical redundancy. Instead of storing a GRS syndrome of the check-pattern lengths of $\bar{\mathbf{x}}=(\mathbf{s}_0,\mathbf{u})\in\Sigma_q^{\ell+k}$, we assign to each possible check-pattern index $j$ a group element $w_j$, and store a single checksum (which we also call a digest) obtained by summing $w_j$ over all symbols belonging to the assigned $j$-th check pattern. When deletions shorten check patterns, the checksum changes by $\sum_{j=1}^{m} \tau_j w_j$, where $m=k+t+1$ is the largest possible number of check patterns in $\bar{\mathbf{x}}$. Therefore, if the sequence $w_1,\ldots,w_m$ is chosen so that all sums of at most $t$ terms are distinct, then this checksum difference uniquely determines the shortened check-pattern indices together with their multiplicities.
	
	We next recall that the notion of Sidon sets, or rather $B_t$ sequences, which is instrumental in implementing the idea above. A sequence $w_1,\ldots,w_m$ in an abelian group $G$ is called a $B_t$ sequence if all multisets of $t$ terms have distinct sums. For the purposes of the code construction here, we will be interested in the so-called $B_{\leq t}$ sequences. Specifically, a sequence $w_1,\ldots,w_m\in G$ is called $B_{\leq t}$ sequence, if all multisets of at most $t$ terms have distinct sums. In other words, for a $B_{\leq t}$ sequence $w_1,\ldots,w_m\in G$, if
	\[
	\sum_{\nu=1}^{\tau} w_{i_\nu}
	=
	\sum_{\nu=1}^{\tau'} w_{j_\nu},
	\qquad
	\tau,\tau'\le t,
	\qquad
	i_{\nu},j_{\nu}\in\{1,\ldots,m\},
	\]
	then $\tau=\tau'$ and the two multisets $\{i_1,\ldots,i_{\tau}\}$ and $\{j_1,\ldots,j_{\tau}\}$	are identical. For fixed $t$, explicit $B_{\leq t}$ sequences of length $m$ exist in groups of size $O(m^t)$. For instance, one may start from an explicit $B_t$ sequence (e.g., the Bose-Chowla construction \cite{bose1962theorems}) and add one coordinate to record the number of summands. 
	\begin{proposition}\label{prop:sidon}
		Let $p$ be a prime power and let $A\subset\mathbb{Z}_{p^t-1}$ be a Bose-Chowla $B_t$ sequence. Define 
		\begin{align*}
			\tilde{A}=\{(\alpha,1)\mid \alpha\in A\}\subset \mathbb{Z}_{p^t-1}\times \mathbb{Z}_{t+1}.
		\end{align*}
		Then $\tilde{A}$ is a $B_{\leq t}$ sequence that has $p$ elements.
	\end{proposition}
	
	\begin{proof}
		By the Bose-Chowla construction, $A$ has $p$ elements. Write $\tilde{A}=\{(\alpha_1,1),\ldots,(\alpha_p,1)\}$.
		Let $0\leq \tau,\tau'\leq t$.  If 
		\begin{align*}
			\sum_{\nu=1}^{\tau} (\alpha_{i_\nu},1) = \sum_{\nu=1}^{\tau'} (\alpha_{j_\nu},1),
		\end{align*} then $\tau=\tau'$ and $\sum_{\nu=1}^{\tau} \alpha_{i_\nu} = \sum_{\nu=1}^{\tau} \alpha_{j_\nu}$. Since a $B_t$ sequence is also a $B_{\tau}$ sequence for $\tau\leq t$. It follows that $\tilde{A}$ is a $B_{\leq t}$ sequence.
	\end{proof}
	
	By Proposition~\ref{prop:sidon}, assuming $p\geq m$, we may take an $m$-subset of $\tilde{A}$ to form the $B_{\leq t}$ sequence $w_1,\ldots,w_m$. Then each $w_j$ is a nonzero element in the abelian group $G:=\mathbb{Z}_{p^t-1}\times \mathbb{Z}_{t+1}$.
	
	We now define the checksum (or digest) used in the construction. 
	For a string $\mathbf{s}$, let $m_{\mathbf{s}}$ be the number of check patterns in $\mathbf{s}$. Let $a_j(\mathbf{s})$ and $e_j(\mathbf{s})$ denote the starting and ending indices of the $j$-th check pattern, respectively. 
	By Lemma~\ref{le:overlap}, any two adjacent check patterns can overlap in at most $\ell-t-1$ indices. In particular, all check patterns have distinct starting indices.
	To avoid double-counting overlaps between adjacent check patterns (as we would like to assign a unique element in $G$ to each check-pattern symbol), define
	\[
	b_j(\mathbf{s})=
	\begin{cases}
		\min\{a_{j+1}(\mathbf{s})-1,e_j(\mathbf{s})\}, & 1\leq j<m_{\mathbf{s}},\\
		e_{m_{\mathbf{s}}}(\mathbf{s}), & j=m_{\mathbf{s}}.
	\end{cases}
	\]
	Thus, the interval $(a_j(\mathbf{s}),a_j(\mathbf{s})+1,\ldots,b_j(\mathbf{s}))$ is indices \emph{assigned} to the $j$-th check pattern, and overlaps are assigned to the later check pattern. In other words, the interval need not contain all indices of the $j$-th check pattern.
%	Define
%	\[
%	f(\mathbf{s},i)=
%	\begin{cases}
%		j, & \text{if }a_j(\mathbf{s})\leq i\leq b_j(\mathbf{s}),\\
%		0, & \text{otherwise}.
%	\end{cases}
%	\]
%	and set $w_0:=0\in G$. 
	Define $B_{\leq t}$-digest of $\mathbf{s}$ to be 
	\begin{align}
		h(\mathbf{s}):=
		%\sum_i w_{f(\mathbf{s},i)}=
		\sum_{j=1}^{m_{\mathbf{s}}}
		\bigl(b_j(\mathbf{s})-a_j(\mathbf{s})+1\bigr)w_j
		\in G.\label{eq:Bt-digest}
	\end{align}
	
	Recall that Lemma~\ref{le:one-to-one} and \ref{le:error} apply to $\bar{\mathbf{x}}$ and its corresponding output string $\bar{\mathbf{y}}$. Therefore, by Lemma~\ref{le:error}, at most $t$ check patterns are shortened in $\bar{\mathbf{y}}$. If $\tau_j$	symbols are deleted from the substring assigned to the $j$-th check pattern, and the corresponding check pattern is not destroyed (i.e., its remaining length is still at least $\ell-t$), then
	\[
	h(\bar{\mathbf{x}})-h(\bar{\mathbf{y}})
	=\sum_j \tau_j w_j,	\qquad	\sum_j \tau_j\leq t.
	\]
	The $B_{\leq t}$-property therefore allows the decoder to recover all multiplicities $\tau_j$. Once these multiplicities are known, the shortened	check patterns are restored by periodic extension.
	As in the GRS construction, the digest $h(\bar{\mathbf{x}})$ is protected by a short $q$-ary deletion-correcting code and appended after $\bar{\mathbf{x}}$. Since the decoding procedure will work with a truncated prefix of the output string (as in the proof of Theorem~\ref{thm:con1}), the last check pattern of $\bar{\mathbf{y}}$ may be destroyed by truncation. We therefore also store the parity of $m_{\bar{\mathbf{x}}}$, allowing the decoder to determine whether the final check pattern is present or missing in the truncated prefix. The resulting construction is given next.

	\begin{construction}
		\label{con:Bt}
		Let $p\geq m$ be a prime power, where $m=k+t+1$ is the largest possible number of check patterns in $\bar{\mathbf{x}}=(\mathbf{s}_0,\mathbf{u})\in\Sigma_q^{k+\ell}$.
		Let $w_1,\ldots,w_m\in G=\mathbb{Z}_{p^t-1}\times \mathbb{Z}_{t+1}$ be a $B_{\leq t}$ sequence constructed by Proposition~\ref{prop:sidon}.
		
		Compute $h(\bar{\mathbf{x}})\in G$ and $\sigma:=m_{\bar{\mathbf{x}}}\bmod 2$. Let $\mathbf{v}$ be a fixed-length $q$-ary representation of $(\sigma,h(\bar{\mathbf{x}}))$. Encode $\mathbf{v}$ using a $q$-ary $t$-deletion-correcting code, for example, one from \cite{sima2020optimal2}, and denote the result by
		$\cR_t(\bar{\mathbf{x}})=\cR_t(\mathbf{s}_0,\mathbf{u})$.
		Given the fixed adapter sequences $\mathbf{s}_0$ and $\mathbf{s}_{n+\ell}$,
		the $t$-deletion-correcting code $\cC$ of length $n$ over $\Sigma_q$ for the nanopore channel is defined to be
		\begin{align*}
			\cC=\{(\mathbf{u},\cR_t(\mathbf{s}_0,\mathbf{u}))\mid \mathbf{u}\in\Sigma_q^k\}.
		\end{align*}
	\end{construction}
	
	\begin{remark}\label{re:Bt}
		Since $|G|=O(m^t)$, the unprotected	digest $\mathbf{v}$ has length $\log_q |G|+O(1)=t\log_q m+O(1)=t\log_q k+O(1)$, where the last equality follows since $m=k+t+1=k+O(1)$.
		In the regime of fixed $q$ and growing $k$, by the result of \cite{sima2020optimal2}, protecting $\mathbf{v}$ by a $q$-ary $t$-deletion-correcting code adds $\Theta(\log\log k)$ symbols. Therefore, the length of $\cR_{t}(\bar{\mathbf{x}})$ is $t\log_q k+\Theta(\log\log k)=t\log_q n+\Theta(\log\log n)$. This matches the smallest possible redundancy up to first order, according to Remark~\ref{re:r-lb}.
	\end{remark}
	
	\begin{theorem}
		\label{thm:Bt}
		Assume $t\leq \min\{(\ell-1)/2,(\ell+2)/3\}$.
		The code $\cC\subset\Sigma_q^n$ described in Construction~\ref{con:Bt} is a $t$-deletion-correcting code for the nanopore channel with redundancy $t\log_q n+\Theta(\log\log n)$.
	\end{theorem}
	
	\begin{proof}
		We describe a decoder that recovers the input string $\mathbf{x}$ from the output $\ell$-mers $\mathbf{z}_0^N$.
		The decoder first applies the mapping $\psi$ to $\mathbf{z}_0^N$ and obtains the output string $\tilde{\mathbf{y}}$. 
		Let $\tau$ be the number of effective symbol deletions in $\tilde{\mathbf{y}}$, and write $\tilde{\mathbf{y}}=(y_{-\ell+1},\ldots,y_{n+\ell-\tau})$. 
		As in Construction~\ref{con:Bt}, the suffix $\cR_{t}(\bar{\mathbf{x}})$ can be recovered from the corresponding subsequence of $\tilde{\mathbf{y}}$, since it is encoded by a $q$-ary $t$-deletion-correcting code. Hence, the decoder obtains the parity $\sigma$ of $m_{\bar{\mathbf{x}}}$ and the $B_{\leq t}$-digest $h(\bar{\mathbf{x}})$ of the string $\bar{\mathbf{x}}$.
		
		Let $\bar{\mathbf{y}}$ be the substring of $\tilde{\mathbf{y}}$ induced by $\bar{\mathbf{x}}$. As discussed before, Lemma~\ref{le:one-to-one} and \ref{le:error} also hold for $\bar{\mathbf{x}}$ and $\bar{\mathbf{y}}$. 
		By Lemma~\ref{le:one-to-one}, check patterns in $\bar{\mathbf{x}}$ and $\bar{\mathbf{y}}$ are matched from left to right, and we have $m_{\bar{\mathbf{x}}}=m_{\bar{\mathbf{y}}}$. Furthermore, by Lemma~\ref{le:error}, at most $\tau\leq t$ check patterns in $\bar{\mathbf{y}}$ are shortened. %Let $\tau_j\geq 0$ be the number of deleted symbols from the $j$-th check pattern in $\bar{\mathbf{x}}$, where $1\leq j\leq m$.
		For $s=0,\ldots,\tau$, define
		\[
		\hat{\mathbf{y}}^{(s)}=(y_{-\ell+1},\ldots,y_{k-\tau+s}).
		\]
		Then $\bar{\mathbf{y}}$ is equal to $\hat{\mathbf{y}}^{(s^*)}$ for some $s^*\in\{0,\ldots,\tau\}$. Define
		\[
		D_s=h(\bar{\mathbf{x}})-h(\hat{\mathbf{y}}^{(s)})\in G.
		\]
		The decoder chooses the smallest $s$ such that $m_{\hat{\mathbf{y}}^{(s)}}\equiv \sigma \pmod{2}$ and that $D_s$ can be represented by a sum of $\tau-s$ elements from the $B_{\leq t}$-sequence $\{w_1,\ldots,w_m\}$. More precisely, for such an $s$, we have $D_s = \sum_{j=1}^m \tau_j w_j\in G$, where $\sum_{j=1}^m\tau_j=\tau-s\leq t$.
		Since $w_1,\ldots,w_m$ is a $B_{\leq t}$ sequence, $D_s$ uniquely determines the multiset
		\[
		\{\underbrace{j,\ldots,j}_{\tau_j\text{ times}}: \tau_j\geq 0\}.
		\]
		%Thus, the decoder determines $\tau_j$ for every $j=1,\ldots,m$. 
		
		We next show that the decoder can restore $\bar{\mathbf{x}}$ from $\hat{\mathbf{y}}^{(s)}$ with the smallest  $s$ that passes the above test for the parity and $B_{\leq t}$-digest. 
		Since $s^*$ passes the test, the smallest passing $s$ satisfies $s\leq s^*$. If $s=s^*$, then $\hat{\mathbf{y}}^{(s)}=\bar{\mathbf{y}}$, and the difference $D_{s}$ is the sum of the number of effective symbols deletions inside each check pattern of $\bar{\mathbf{x}}$. Thus, $\bar{\mathbf{x}}$ can be recovered by extending each shortened check pattern in $\bar{\mathbf{y}}$ by $\tau_j$ symbols according to its periodicity.
		
		Now assume $s<s^*$. Then $\hat{\mathbf{y}}^{(s)}$ is obtained from $\bar{\mathbf{y}}$ by deleting a suffix of length $s^*-s\leq t$. Since check patterns have length at least $\ell-t>t$, this truncation can destroy only the last check pattern. The parity test
		$m_{\mathbf{y}^{(s)}}\equiv \sigma\pmod{2}$ rules out the case in which the last check pattern is destroyed. Thus, all check patterns are still present in $\mathbf{y}^{(s)}$.
		Since $s^*-s\leq t<\ell-t$, the omitted length-$(s^*-s)$ suffix either lies entirely inside the last check pattern or contains at least one symbol outside the last check pattern (and hence all check patterns). According to \eqref{eq:Bt-digest}, symbols outside check patterns contribute zero to the $B_{\leq t}$-digest. If any omitted symbol contributed zero, then $D_s$ would be a sum of $\sum_{j=1}^m\tau_j<\tau-s$ elements of the $B_{\leq t}$ sequence, which is a contradiction. Therefore, for a passing $s<s^*$, every omitted symbol lies inside the last check pattern. Therefore, periodic extension based on $\{\tau_j:j=1,\ldots,m\}$ restores the omitted suffix as well as the deletions, yielding $\bar{\mathbf{x}}$.

		Hence, in all cases the decoder recovers $\bar{\mathbf{x}}$ and thus also $\mathbf{x}$.
	\end{proof}
	
	\begin{remark}
		Construction~\ref{con:Bt} is explicit for fixed $t$. Indeed, the check patterns of $\bar{\mathbf{x}}$ can be found in linear time for fixed $\ell,t$. The $B_{\leq t}$ sequence can be constructed explicitly from the Bose-Chowla construction and Proposition~\ref{prop:sidon} over a group of size $O(m^t)=O(n^t)$. The sequence
		$(w_1,\ldots,w_m)$, as well as a lookup table mapping every sum of at most $t$ terms to the corresponding multiset of indices, can be precomputed in $O(m^t)$ time and space. Once this preprocessing is done, encoding requires time linear in $n$. Decoding consists of the output-string and check-pattern reconstruction steps as before, together with one lookup in the $B_{\le t}$ table. Thus, for constant $t,\ell,q$, the encoder and decoder run in polynomial time; with the $B_{\le t}$ table precomputed, the complexity is essentially linear in $n$, apart from the $\mathrm{polylog}(n)$-time encoding/decoding of the short redundancy block.
	\end{remark}
	
	As a comparison, we note that the $B_{\leq t}$-based construction naturally uses a lookup table of size $O(n^t)$ if one wants fast decoding of the checksum difference. In contrast, the GRS-based construction can be decoded using variants of GRS decoders without precomputing a lookup table.
	
	\section{Edit-correcting codes}\label{sec:edit}
	We now consider the edit version of the nanopore channel, in which the channel may introduce insertions, deletions, and substitutions of $\ell$-mers. We continue to assume that the two adapter $\ell$-mers are known and anchored.
	
	Let $\isi_{\ell}(\tilde{\mathbf{x}})=(\mathbf{s}_0,\ldots,\mathbf{s}_{n+\ell})$	denote the $\ell$-mer sequence generated from the input string $\tilde{\mathbf{x}}=(\mathbf{s}_0,\mathbf{x},\mathbf{s}_{n+\ell})$.
	Write $\de(\cdot,\cdot)$ for the edit distance between two $\ell$-mer sequences, where insertions, deletions, and substitutions of $\ell$-mers each have cost one.
	
	We begin with a simple reduction that shows that any code correcting $2t$ $\ell$-mer deletions also corrects $t$ $\ell$-mer edits. This is the nanopore analogue of the standard deletion-to-edit reduction.
	
	\begin{proposition}
		\label{prop:del-to-edit}
		If a code $\cC\subset \Sigma_q^n$ corrects $2t$ deletions of $\ell$-mers in the nanopore channel, then it also corrects $t$ edits of $\ell$-mers.
	\end{proposition}
	
	\begin{proof}
		Toward a contradiction, suppose that two distinct codewords $\mathbf{x},\mathbf{x}'\in\cC$ can produce the same output $\ell$-mer sequence $\mathbf{z}_0^N$ after at most $t$ edits each. By the triangle inequality for edit distance, the corresponding input strings $\tilde{\mathbf{x}},\tilde{\mathbf{x}}'$ then satisfy 
		\[
		\de\big(\isi_\ell(\tilde{\mathbf{x}}),\isi_\ell(\tilde{\mathbf{x}}')\big)
		\leq 2t.
		\]
		Replacing each substitution by one deletion and one insertion leads to at most $4t$ insertions and deletions in total. Since the two $\ell$-mer sequences $\isi_\ell(\tilde{\mathbf{x}}),\isi_\ell(\tilde{\mathbf{x}}')$ have the same length $n+\ell+1$, this implies that their longest common subsequence has length at least $n+\ell+1-2t$. Equivalently, the two codewords have a common $\ell$-mer subsequence obtainable	from each by deleting at most $2t$ $\ell$-mers. This contradicts the assumption that $\cC$ corrects $2t$ $\ell$-mer deletions.
	\end{proof}
	
	By Proposition~\ref{prop:del-to-edit}, one may obtain $t$-edit-correcting codes for the nanopore channel by applying any $2t$-deletion-correcting code. As a consequence, Construction~\ref{con:Bt} gives $t$-edit-correcting codes with redundancy $2t\log_q n+\Theta(\log\log n)$. This black-box reduction is useful, but it leads to larger redundancy than necessary, as $t$-edit-correcting codes are not necessary $2t$-deletion-correcting codes. We next give a direct construction that exploits the extra structure of $\ell$-mer edit ambiguities.
	
	\subsection{A direct construction using new check patterns}
	
	For the structural arguments below, we first recall the notion of an \emph{edit alignment} between two sequences of the same length. Let $\mathbf{a}=(a_1,\ldots,a_n)$ and $\mathbf{b}=(b_1,\ldots,b_n)$ be two sequences over an alphabet $\cA$. An edit alignment of $\mathbf{a}$ and $\mathbf{b}$ is a pair of equal-length sequences $\mathbf{a}'=(a'_1,\ldots,a'_M),\mathbf{b}'=(b'_1,\ldots,b'_M)$ over the alphabet $\cA\cup\{*\}$ such that deleting all gap symbols $*$ from $\mathbf{a}'$ and $\mathbf{b}'$ gives $\mathbf{a}$ and $\mathbf{b}$, respectively, and $(a'_i,b'_i)\neq(*,*)$ for all $i$. A \emph{column} $(a'_i,b'_i)$ has cost $0$ if $a'_i=b'_i\in\cA$, and has cost $1$ otherwise. Thus, a positive-cost column with two non-gap symbols is a substitution, while a positive-cost column containing one gap symbol is an insertion/deletion. The edit distance is the minimum total cost over all edit alignments.
	
	A zero-cost column is called a matched column. A maximal consecutive block of positive-cost columns, bounded on both sides by matched columns, will be called an \emph{alignment edit burst}. In our discussion below, the two sequences are ISI outputs, so the matched boundary columns consist of equal $\ell$-mers. Since we assume all ISI output $\ell$-mer sequences start with the same left adapter $\ell$-mer and end with the same right adapter $\ell$-mer, we may choose an optimal alignment in which these two adapter $\ell$-mers are matched. 
	
	The direct construction below follows the same philosophy as the deletion-correcting constructions, except for a different defining length of the check patterns.
	
	\begin{definition}[Edit check patterns]
		\label{def:edit-check}
		A substring of $\mathbf{s}=(s_1,\ldots,s_n)\in\Sigma_q^n$ is a check pattern if it is a maximal periodic pattern with period at most $2t$ and length at least \(\ell-2t\).
	\end{definition}
	
	We assume throughout this subsection that $\ell-2t\geq \max\{2t+1,4t-2\}$, i.e., $t\leq \min\{(\ell-1)/4,(\ell+2)/6\}$.
	The choice of edit check-pattern length is forced by the following intuition. If the $\ell$-mer sequences $\isi_\ell(\tilde{\mathbf{x}}),\isi_\ell(\tilde{\mathbf{x}}')$ of two candidate codewords both lie within $t$ edits of the output $\ell$-mer sequence, then $\isi_\ell(\tilde{\mathbf{x}})$ and $\isi_\ell(\tilde{\mathbf{x}}')$ are within $2t$ edits of each other. 
	Consider an alignment edit burst between two matched boundary $\ell$-mers. If these two boundary $\ell$-mers are shifted by gaps $g_1$ and $g_2$ in the two sequences, then the cost of the burst is at least $\max\{g_1-1,g_2-1\}$. Since the edit distance is at most $2t$, the cost is at most $2t$. Thus, $g_1,g_2\leq 2t+1$.
	If $g_1=g_2$, by Lemma~\ref{le:consistent-edit}, the burst is trivial, and every nontrivial burst has $g_1\ne g_2$. Writing \(g_1<g_2\), the two gaps force the overlap corresponding to the smaller gap to have period $g_2-g_1\leq 2t$ and length at least $\ell-g_1\geq\ell-g_2+1\geq \ell-2t$.
	Therefore, every local edit ambiguity is witnessed by a periodic string with period at most $2t$ and length at least $\ell-2t$. This motivates defining edit check patterns as maximal periodic patterns with period at most
	$2t$ and length at least $\ell-2t$. 
	The condition $\ell-2t\geq 2t+1$ ensures that such a pattern is long enough to reveal a	period as large as $2t$, while the condition $\ell-2t\geq 4t-2$ ensures,	via Theorem~\ref{thm:FW}, that adjacent edit check patterns overlap in at most $\ell-2t-1$ symbols. Since every edit check pattern has length at least $\ell-2t$, no edit check pattern can contain another. In particular, all edit check patterns in a string start at distinct positions, and thus they can be indexed unambiguously from left to right.
	
	The construction below is obtained by modifying Construction~\ref{con:Bt} so as	to store a digest of the edit check patterns. By a slight abuse of notation, we	write $m_{\mathbf{s}}$ for the number of edit check patterns in a string
	$\mathbf{s}$, and we write $m$ for the maximum possible number of edit check patterns in $\bar{\mathbf{x}}=(\mathbf{s}_0,\mathbf{u})$.
	
	\begin{construction}
		\label{con:edit-Bt}
		Let $p\geq m$ be a prime power, where $m=k+2t+1$ is the largest possible number of edit check patterns in $\bar{\mathbf{x}}=(\mathbf{s}_0,\mathbf{u})\in\Sigma_q^{k+\ell}$.
		Let $w_1,\ldots,w_m\in G=\mathbb{Z}_{p^t-1}\times \mathbb{Z}_{t+1}$ be a $B_{\leq t}$ sequence constructed by Proposition~\ref{prop:sidon}.
		
		Compute $h(\bar{\mathbf{x}})\in G$ and $\sigma:=m_{\bar{\mathbf{x}}}\bmod 2$. Let $\mathbf{v}$ be a fixed-length $q$-ary representation of $(\sigma,h(\bar{\mathbf{x}}))$. Encode $\mathbf{v}$ using a $q$-ary $t$-edit-correcting code, for example, a $2t$-deletion-correcting code from \cite{sima2020optimal2}, and denote the result by
		$\cR_{2t}(\bar{\mathbf{x}})=\cR_{2t}(\mathbf{s}_0,\mathbf{u})$.
		Given the fixed adapter sequences $\mathbf{s}_0$ and $\mathbf{s}_{n+\ell}$,
		the $t$-edit-correcting code $\cC$ of length $n$ over $\Sigma_q$ for the nanopore channel is defined to be
		\begin{align*}
			\cC=\{(\mathbf{u},\cR_{2t}(\mathbf{s}_0,\mathbf{u}))\mid \mathbf{u}\in\Sigma_q^k\}.
		\end{align*}
	\end{construction}
	
	\begin{remark}
		The redundancy of the codes in Construction~\ref{con:edit-Bt} can be estimated similarly as in Remark~\ref{re:Bt}. 
		Since $t,\ell,q$ are fixed, for growing $k$, the length of $\cR_{2t}(\bar{\mathbf{x}})$ equals $t\log_q n+\Theta(\log\log n)$. Thus, the redundancy of Theorem~\ref{thm:edit-Bt} is optimal to first order. Indeed,	the $t$-edit nanopore channel contains the $t$-deletion nanopore channel as a special case. Hence, any code that corrects $t$ edits of $\ell$-mers must correct $t$ deletions of $\ell$-mers. By Remark~\ref{re:r-lb}, every $t$-deletion-correcting code for the nanopore channel has redundancy at least $t\log_q n+\Omega(1)$.
		Therefore, every $t$-edit-correcting code satisfies the same lower bound. It follows that the redundancy of the codes in Construction~\ref{con:edit-Bt} matches the lower bound in the leading term and is first-order optimal.
	\end{remark}	
	
	To prove that Construction~\ref{con:edit-Bt} corrects $t$ edits, we first	record two technical lemmas. 
	The first one relates the edit distance between $\ell$-mer sequences to the edit distance between the corresponding strings that generate them. 
	In general, a single edit in an $\ell$-mer sequence may correspond to several symbol edits in the underlying string. However, when the $\ell$-mer edit distance is small, each $\ell$-mer	edit can be charged to at most one $q$-ary symbol edit. Throughout, we denote the edit distance between two $q$-ary strings by $\de^{(q)}(\cdot,\cdot)$.
	
	\begin{lemma}
		\label{le:lmer-to-symbol-edit}
		Let $\tilde{\mathbf{x}},\tilde{\mathbf{x}}'\in\Sigma_q^{n+2\ell}$ be two input strings with the same anchored adapter $\ell$-mers. Let $D\leq\ell-1$ be an integer. If
		$
		\de
		\big(
		\isi_\ell(\tilde{\mathbf{x}}),
		\isi_\ell(\tilde{\mathbf{x}}')
		\big)
		\leq D,
		$
		then
		$
		\de^{(q)}(\tilde{\mathbf{x}},\tilde{\mathbf{x}}')\leq D.
		$
	\end{lemma}
	
	\begin{proof}
		Take an optimal edit alignment between $\isi_\ell(\tilde{\mathbf{x}})$ and $\isi_\ell(\tilde{\mathbf{x}}')$ in which the two adapter $\ell$-mers are matched. Decompose the alignment into alignment edit bursts.
		Consider one such burst, bounded by the same two matched boundary $\ell$-mers. Assume that the burst  contains $\tau$ \(\ell\)-mers between the boundary $\ell$-mers in $\isi_\ell(\tilde{\mathbf{x}})$ and $\rho$ $\ell$-mers between the boundaries in $\isi_\ell(\tilde{\mathbf{x}}')$. The cost of this burst is at least $\max\{\tau,\rho\}$. Since the edit distance is at most $D\leq\ell-1$, we have $\tau,\rho\leq \ell-1$.
		
		If $\tau=\rho$, then Lemma~\ref{le:consistent-edit} implies that the burst is trivial: the two local $\ell$-mer sequences between the boundary $\ell$-mers are identical. Thus, this burst contributes no $q$-ary symbol edits.
		
		If $\tau\ne\rho$, then the two boundary $\ell$-mers are shifted by different gaps $\tau+1$ and $\rho+1$, and the two gaps give rise to two corresponding local $q$-ary strings of different length. Moreover, the longer one is obtained from the shorter one by inserting $|\tau-\rho|$ symbols in the periodic pattern described in Lemma~\ref{le:consistent-edit}. Hence, the two local $q$-ary strings differ by at most $|\tau-\rho|$ $q$-ary symbol insertions or deletions. Since $|\tau-\rho|\leq \max\{\tau,\rho\}$, the $q$-ary edit cost of this local burst is at most its $\ell$-mer edit cost.
		Summing over all alignment edit bursts gives $\de^{(q)}(\tilde{\mathbf{x}},\tilde{\mathbf{x}}')\leq D$.
	\end{proof}
	
	Although edit distance is not generally monotone under coordinate restriction, the next lemma shows that edit distance between two equal-length strings is monotone under restriction to a fixed contiguous coordinate window.
	
	\begin{lemma}
		\label{le:fixed-block-edit}
		Let $\tilde{\mathbf{x}},\tilde{\mathbf{x}}'\in\Sigma_q^{n+2\ell}$ with $\tilde{\mathbf{x}}=(\mathbf{s}_0,\mathbf{u},\mathbf{r},\mathbf{s}_{n+\ell})$ and $\tilde{\mathbf{x}}'=(\mathbf{s}_0,\mathbf{u}',\mathbf{r}',\mathbf{s}_{n+\ell})$, where $\mathbf{u},\mathbf{u}'\in\Sigma_q^k$ and $\mathbf{r},\mathbf{r}'\in\Sigma_q^{n-k}$.
		Then
		\[
		\de^{(q)}(\mathbf{r},\mathbf{r}')\leq\de^{(q)}(\tilde{\mathbf{x}},\tilde{\mathbf{x}}'),\qquad 
		\de^{(q)}((\mathbf{s}_0,\mathbf{u}),(\mathbf{s}_0,\mathbf{u}'))\leq\de^{(q)}(\tilde{\mathbf{x}},\tilde{\mathbf{x}}').
		\]
	\end{lemma}
	
	\begin{proof}
		For this proof, we index the coordinates of $\tilde{\mathbf{x}}$ and $\tilde{\mathbf{x}}'$ starting from $1$.
		Take an optimal edit alignment of $\tilde{\mathbf{x}}$ and $\tilde{\mathbf{x}}'$. Since the two strings have the same length $n+2\ell$, the number of deletions equals the number of insertions. Denote this number by $\delta$, and let $s$ be the number of substitutions. Then $s=\de^{(q)}(\tilde{\mathbf{x}},\tilde{\mathbf{x}}')-2\delta$.
		
		The alignment pairs $(n+2\ell)-\delta$ symbols of $\tilde{\mathbf{x}}$ with $(n+2\ell)-\delta$ symbols of
		$\tilde{\mathbf{x}}'$ in an order-preserving way. Write these paired coordinates as
		\[
		i_1<i_2<\cdots<i_{n+2\ell-\delta},\qquad
		j_1<j_2<\cdots<j_{n+2\ell-\delta},
		\]
		where $i_{\nu}$ is paired with $j_{\nu}$.
		Let $I=\{a,\ldots,b\}$ be the fixed coordinate interval such that the restriction of $\tilde{\mathbf{x}}$ to $I$ is $\mathbf{r}$, i.e., $\tilde{\mathbf{x}}_I=\mathbf{r}$. Then $\tilde{\mathbf{x}}'_I=\mathbf{r}'$. 
		If $|I|\leq \delta$, then
		\[
		\de^{(q)}(\mathbf{r},\mathbf{r}')
		\leq |I|
		\leq 2\delta+s
		=
		\de^{(q)}(\tilde{\mathbf{x}},\tilde{\mathbf{x}}'),
		\]
		so the claim is immediate.
		
		Assume now that $|I|>\delta$. For every $\nu\in\{a,a+1,\dots,b-\delta\}$, we have
		\[
		a\leq \nu\leq i_\nu\leq \nu+\delta\leq b,
		\qquad
		a\leq \nu\leq j_\nu\leq \nu+\delta\leq b.
		\]
		Thus, the alignment contains at least $b-\delta-a+1=|I|-\delta$ paired columns whose two coordinates both lie in $I$.
		Restrict the alignment to these paired columns lying in $I\times I$, keeping their matched and substitution columns, and align all remaining symbols of $\mathbf{r}$ and $\mathbf{r}'$ to gaps. Consider the resulting alignment of $\mathbf{r}$ and $\mathbf{r}'$. 
		Since at least $|I|-\delta$ symbols	on each side of this alignment are paired inside $I\times I$, at most $\delta$ symbols remain
		unpaired on each side. Moreover, the number of substitution columns retained is at most $s$.
		Therefore,
		\[
		\de^{(q)}(\mathbf{r},\mathbf{r}')
		=\de^{(q)}(\tilde{\mathbf{x}}_I,\tilde{\mathbf{x}}'_I)
		\leq 2\delta+s
		= \de^{(q)}(\tilde{\mathbf{x}},\tilde{\mathbf{x}}').
		\]
		Lastly, $\de^{(q)}((\mathbf{s}_0,\mathbf{u}),(\mathbf{s}_0,\mathbf{u}'))\leq\de^{(q)}(\tilde{\mathbf{x}},\tilde{\mathbf{x}}')$ follows by taking $I$ to be such that $\tilde{\mathbf{x}}_I=(\mathbf{s}_0,\mathbf{u})$.
	\end{proof}
	
	\begin{theorem}
		\label{thm:edit-Bt}
		Assume $t\leq \min\{(\ell-1)/4,(\ell+2)/6\}$.
		The code $\cC$ in Construction~\ref{con:edit-Bt} is a $t$-edit-correcting code for the
		nanopore channel with redundancy $t\log_q n+O(\log\log n)$.
	\end{theorem}
	
	\begin{proof}
		Given an output $\ell$-mer sequence $\mathbf{z}_0^N$, 
		the decoder enumerates all subsequences obtained by discarding at most $t$ output $\ell$-mers, and keeps only those subsequences whose length lies between $n+\ell+1-t$ and $n+\ell+1$ and whose first and last $\ell$-mers are the two adapters $\mathbf{s}_0$ and $\mathbf{s}_{n+\ell}$.		
		For each retained subsequence, the decoder applies the mapping $\psi$ and keeps the resulting string only if its length lies between $n+2\ell-t$ and $n+2\ell$. It then invokes the $t$-deletion decoder (as in the proof Theorem~\ref{thm:Bt}) with edit check patterns in place of check patterns. This produces a candidate input string, denoted by $\tilde{\mathbf{x}}'=(\mathbf{s}_0,\mathbf{x}',\mathbf{s}_{n+\ell})$. Finally, the decoder verifies whether
		\[
		\de\big(\isi_\ell(\tilde{\mathbf{x}}'),	\mathbf{z}_0^N\big)\leq t.
		\]
		Candidates failing this verification are discarded.
		The true input string $\tilde{\mathbf{x}}$ (and hence the true codeword $\mathbf{x}$) appears in the enumeration, since it can be obtained by decoding the subsequence of $\mathbf{z}_0^N$ with the inserted and substituted output $\ell$-mers discarded.
		
		It remains to show that at most one candidate passes the edit distance verification. Suppose two candidates $\tilde{\mathbf{x}}=(\bar{\mathbf{x}},\cR_{2t}(\bar{\mathbf{x}}),\mathbf{s}_{n+\ell})$ and $\tilde{\mathbf{x}}'=(\bar{\mathbf{x}}',\cR_{2t}(\bar{\mathbf{x}}'),\mathbf{s}_{n+\ell})$ both pass. Since both are within $t$ edits of $\mathbf{z}_0^N$, the triangle inequality gives 
		\[
		\de\big(\isi_\ell(\tilde{\mathbf{x}}),\isi_\ell(\tilde{\mathbf{x}}')\big)\leq 2t.
		\]	
		Under the assumption $t\leq (\ell-1)/4$, we have $2t<\ell$. Thus, Lemma~\ref{le:lmer-to-symbol-edit} and Lemma~\ref{le:fixed-block-edit} imply that the redundancy blocks satisfy
		\[
		\de^{(q)}(\cR_{2t}(\bar{\mathbf{x}}),\cR_{2t}(\bar{\mathbf{x}}'))\leq 2t.
		\]
		Since these redundancy blocks are codewords of a $q$-ary $t$-edit-correcting code, we have $\cR_{2t}(\bar{\mathbf{x}})=\cR_{2t}(\bar{\mathbf{x}}')$ and thus 
		\[
		h(\bar{\mathbf{x}})=h(\bar{\mathbf{x}}').
		\]

		We now compare the two prefixes $\bar{\mathbf{x}}$ and $\bar{\mathbf{x}}'$.
		Consider an edit alignment between $\isi_\ell(\bar{\mathbf{x}})$ and $\isi_\ell(\bar{\mathbf{x}}')$ of cost at most $2t$, obtained by restricting an optimal edit alignment between $\isi_\ell(\tilde{\mathbf{x}})$ and $\isi_\ell(\tilde{\mathbf{x}}')$ to the prefix portions.
		Decompose the alignment into edit alignment bursts.	Each edit alignment burst between two matched boundary $\ell$-mers connects the same endpoints with two
		different feasible gaps. Since the gap is at most $2t+1<\ell$, if the gaps are equal, then by Lemma~\ref{le:consistent} the intermediate $\ell$-mers are uniquely determined and the burst would be trivial. 
		Thus, the gaps should be distinct. By Lemma~\ref{le:consistent} again, they force a periodic pattern of period at most $2t$ and length at least $\ell-2t$. Hence, in this case, the burst changes the length of one edit check pattern. As mentioned earlier, the condition $t\leq (\ell+2)/6$ ensures that one burst cannot ambiguously merge or split
		several edit check patterns.
		
		Let $\tau_j\in\mathbb{Z}$ be the signed change in the length of the $j$-th edit check pattern when passing from $\bar{\mathbf{x}}$ to $\bar{\mathbf{x}}'$. Since the two strings $\bar{\mathbf{x}},\bar{\mathbf{x}}'$ have the same length, by Lemma~\ref{le:lmer-to-symbol-edit} and \ref{le:fixed-block-edit}, the edit distance is at most $2t$. It follows that the total positive and negative changes satisfy
		\[
		\sum_{\tau_j>0} \tau_j\leq t,
		\qquad
		\sum_{\tau_j<0} (-\tau_j)\leq t.
		\]
		Moreover,
		\[
		0=h(\bar{\mathbf{x}}')-h(\bar{\mathbf{x}})
		=
		\sum_j \tau_j w_j.
		\]
		It follows that
		\begin{align}
		\sum_{\tau_j>0} \tau_j w_j
		=
		\sum_{\tau_j<0} (-\tau_j)w_j.\label{eq:net-change}
		\end{align}
		Each side of \eqref{eq:net-change} is a sum of at most $t$ terms from the $B_{\leq t}$ sequence.
		Therefore, by the $B_{\leq t}$-property, the two multisets of indices are identical, and all $\tau_p$ must be zero. Thus, no edit check pattern changes length. In other words, no nontrivial edit alignment burst occurs. Hence, $\bar{\mathbf{x}}=\bar{\mathbf{x}}'$.
	\end{proof}
	
	\begin{remark}
		The decoder described in the proof Theorem~\ref{thm:edit-Bt} enumerates ${O(n^t)}$ subsequence and therefore runs in polynomial time for fixed $t$.
		For each candidate, the edit-distance verification can be done by a dynamic program in $O(n)$ time. Thus, the codes in Construction~\ref{con:edit-Bt} is polynomial-time decodable for every fixed $t,\ell,q$.
	\end{remark}
	
	\subsection{A general framework}
	
	Construction~\ref{con:edit-Bt} and Theorem~\ref{thm:edit-Bt} suggest a general check-pattern framework for edit correction in the nanopore channel. The main components are candidate generation and distance verification.

	\paragraph{Candidate generation}
		Given the output $\ell$-mer sequence, the decoder first enumerates possible ways to discard the output $\ell$-mers that may have arisen from insertions or substitutions. For each retained subsequence with a feasible length, the decoder treats it as a deletion-only output and applies the corresponding deletion-correcting decoder.
		This produces a list of candidate input strings. The true input string is guaranteed to appear in this list, since deleting all inserted output $\ell$-mers and all substituted output $\ell$-mers leaves a deletion-only
		output of the true $\ell$-mer sequence.
		
	\paragraph{Distance verification}
		The construction stores a digest of the edit-check-pattern information, protected by a $q$-ary $t$-edit correcting code. The decoder outputs a candidate input string if its $\ell$-mer sequence is within $t$ edits of the output $\ell$-mer sequence.
	
	The framework is specific to the small-$t$ regime. When $\ell-2t<\max\{2t+1,4t-2\}$, the edit check patterns may be too short, and distinct patterns may overlap substantially. In that regime, a single alignment edit burst may merge, split, or obscure several check patterns, so the ordered check-pattern description is no longer possible and unique decoding is not guaranteed.
	
	The choice of digest in the verification step is flexible, and different choices of digest give different constructions. For instance, Construction~\ref{con:rs} also extends to the $t$-edit nanopore channel.	
	The only structural changes are that the check patterns must be defined to have length $\ell-2t$, i.e., replaced by edit check patterns, and that the GRS syndrome should be protected by a $q$-ary $t$-edit-correcting code.

	To see why this replacement in Construction~\ref{con:rs} gives unique decoding, suppose that two candidates both pass
	the verification step. Then their $\ell$-mer sequences are within $2t$ edits of each other. By Lemmas~\ref{le:lmer-to-symbol-edit} and	\ref{le:fixed-block-edit}, the two protected redundancy blocks have $q$-ary edit distance at most $2t$. Since the redundancy block is encoded by a $t$-edit-correcting code, the two candidates must encode the same GRS syndrome. On the other hand, under the condition $\ell-2t\geq \max\{2t+1,4t-2\}$, an edit-alignment argument shows that the edit check vectors of the two candidates differ in at most $t$ positions. Since the underlying GRS code has minimum distance $2t+1$, equality of the syndromes implies that the two check vectors should be equal, and hence the two candidates must be identical.
	Thus the GRS-based digest also yields a $t$-edit-correcting construction, with redundancy $2t\log_q n+\Theta(\log\log n)$.

	\section{Concluding remarks}\label{sec:re}
	We studied error-correcting codes for an adversarial nanopore channel in which a $q$-ary string is first mapped to a sequence of $\ell$-mers and the adversary then corrupts this $\ell$-mer sequence with a budget of $t$ edits. A central theme of the paper is that ambiguity in this channel is highly
	structured: whenever different consistent output $\ell$-mers exist, the underlying string must contain a periodic pattern. This observation leads naturally to the notion of check patterns. In the small-$t$ regimes considered in this paper, Theorem~\ref{thm:FW} ensures that check patterns cannot overlap too much, so they can be ordered and matched from left to right. This reduces reconstruction to an almost Hamming-type error-correction problem on check-pattern information, since the channel affects only a small number of such patterns. It is therefore natural to ask what can be done beyond the small-$t$ regime. It would be interesting to develop richer synchronization digests, perhaps recording additional period or positional information, that work for larger values of $t$ relative to $\ell$.
	
	The deletion-only constructions presented in this paper both admit polynomial-time encoding/decoding for fixed $t,\ell$, and $q$. More specifically, Construction~\ref{con:rs} admits linear-time decoding, whereas Construction~\ref{con:Bt} has linear online decoding only if a lookup table for the $B_{\leq t}$ sums is precomputed. It would therefore be interesting to construct first-order optimal deletion-correcting codes for the nanopore channel that admit linear-time decoding without such preprocessing.
	
	For the edit channel, our decoder uses an enumerate-and-verify strategy: it enumerates possible ways to discard inserted or substituted output $\ell$-mers, runs a deletion-correcting decoder on the remaining subsequence, and verifies the resulting candidate. This results in polynomial time complexity for fixed $t$, but grows like $n^{O(t)}$. Designing faster, possibly linear, decoders for the edit channel is another natural direction for future work.
	
	Lastly, the $B_{\leq t}$-based constructions are first-order optimal, but they still have $\Theta(\log\log n)$ overhead arising from protecting a short $\Theta(\log n)$-symbol digest. It remains open whether this overhead can be reduced to $O(1)$, or whether a nontrivial lower-order term is inherent for explicit constructions in this channel model.

	% \balance
	\bibliographystyle{IEEEtran}
	\bibliography{ref}
\end{document}